\documentclass[journal]{IEEEtran}
\usepackage{graphicx}
\usepackage{amsmath}
\usepackage{array}
\usepackage{mdwmath}
\usepackage{amssymb}
\usepackage{mdwtab}
\usepackage{stfloats}
\usepackage[tight,footnotesize]{subfigure}
\usepackage{amsmath,amsthm}
\usepackage{threeparttable}
\usepackage{color}
\usepackage{url}
\usepackage{algpseudocode}
\usepackage{algorithm}
\usepackage{multicol}
\usepackage{amssymb}
\usepackage{epstopdf}
\usepackage[colorlinks]{hyperref}

\algrenewcommand{\algorithmicrequire}{\textbf{Input:}}
\algrenewcommand{\algorithmicensure}{\textbf{Output:}}

\newtheorem{corollary}{Corollary}
\newtheorem{proposition}{Proposition}

\newtheorem{theorem}{Theorem}
\newtheorem{remark}{Remark}

\def\BibTeX{{\rm B\kern-.05em{\sc i\kern-.025em b}\kern-.08em
    T\kern-.1667em\lower.7ex\hbox{E}\kern-.125emX}}
\begin{document}
\title {
Fluid Antenna Systems Enabling 6G HRLLC\\With Port Switching Delay
}
\author{
Xusheng Zhu,
Kai-Kit Wong, \IEEEmembership{Fellow,~IEEE},
Hao Xu,
Chenguang Rao, and
Hyundong Shin, \IEEEmembership{Fellow, IEEE}
\vspace{-9mm}


\thanks{X. Zhu, K. K. Wong, and C. Rao are with the Department of Electronic and Electrical Engineering, University College London, London, United Kingdom. K. K. Wong is also affiliated with the Department of Electronic Engineering, Kyung Hee University, Yongin-si, Gyeonggi-do 17104, Korea (e-mail: \{xusheng.zhu; kai-kit.wong; chenguang.rao\}@ucl.ac.uk).}
\thanks{H. Xu is with the National Mobile Communications Research Laboratory, Southeast University, Nanjing 210096, China (e-mail: hao.xu@seu.edu.cn)}
\thanks{H. Shin is with the Department of Electronics and Information Convergence Engineering, Kyung Hee University, Yongin-si, Gyeonggi-do 17104, Republic of Korea (e-mail: hshin@khu.ac.kr).}

\thanks{(\emph{Corresponding author: Kai-Kit Wong}).}
}

\maketitle
\begin{abstract}
Fluid antenna systems (FAS) exploit antenna position reconfigurability to harvest spatial diversity within compact form factors, making them a promising enabler for sixth-generation (6G) user terminals (UTs). However, practical port switching introduces latency and signaling overhead, which can be particularly detrimental to hyper-reliable low-latency communications (HRLLC) under finite blocklength operation. This paper investigates FAS-enabled HRLLC by jointly accounting for spatial correlation, port switching delay, and finite blocklength coding.
To obtain tractable performance characterization, we adopt an effective spatial degrees-of-freedom (DoF) representation of the correlated FAS channel and derive closed-form analytical approximations for the average block error rate (BLER) and average achievable rate. The analysis reveals a fundamental design trade-off: increasing the number of ports improves spatial diversity but reduces the effective blocklength, thereby intensifying the finite-blocklength penalty. We further analyze the structural reliability, rate, and energy-efficiency trends with respect to the port dimension and obtain the discrete optimal port configuration through finite search over the feasible integer set. In addition, a switching-delay threshold is characterized to identify the operating regime in which FAS can provide performance gains over a fixed-position antenna (FPA) system. Numerical results validate the analytical framework and show that substantial HRLLC performance gains are achievable when the port switching delay is sufficiently small.
\end{abstract}

\begin{IEEEkeywords}
Fluid antenna system (FAS), 6G, hyper-reliable low-latency communications (HRLLC), finite blocklength, port switching delay, spatial diversity.
\end{IEEEkeywords}

\vspace{-2mm}
\section{Introduction}\label{sec:introduction}
\IEEEPARstart{T}{he advent} of sixth-generation (6G) wireless is poised to catalyze a paradigm shift in global connectivity, moving from conventional human-centric communications toward machine-centric ecosystems. This shift will support mission-critical applications with unprecedented stringency, ranging from autonomous industrial robotics and remote telesurgery to immersive holographic presence and tactile internet, e.g., \cite{haq2023a, pour2024av}. These futuristic applications rely heavily on the overarching framework of hyper-reliable low-latency communications (HRLLC), imposing constraints on system performance, typically mandating packet error rates as low as $10^{-5}$ to $10^{-9}$ within sub-millisecond latency budgets \cite{zhou2024opt}.

To satisfy these exacting standards, mitigating the detrimental effects of deep channel fading is paramount. Traditionally, multiple-input multiple-output (MIMO) systems have been the cornerstone for enhancing link reliability \cite{zhu2025tris,zhu2025per}. However, a fundamental bottleneck arises in the hardware design of future 6G edge devices: the physical space constraints of compact terminals, e.g., internet-of-thing (IoT) sensors, wearables, and nanobots, etc., severely hamper the deployment of multiple fixed-position antennas (FPAs) with sufficient spacing between them. This spatial restriction caps the achievable diversity gain, rendering conventional MIMO solutions insufficient for the extreme reliability requirements of 6G HRLLC.

To overcome this physical barrier, the fluid antenna system (FAS) has emerged as a refreshing technology that redefines the concept of spatial diversity \cite{wong2021flui, wong2020per}. Unlike traditional FPAs which are bound to discrete, static locations, FAS treats the antenna as a reconfigurable physical-layer resource to broaden system design and network optimization \cite{I35_New2024aTutorial,new2025flar,hong2025ac}. While FAS is hardware agnostic, implementation examples include liquid-based structures \cite{I24_shen2024design,wu2025fas}, movable elements \cite{I30_zhu2024historical,zjpi2024ma}, and reconfigurable pixels \cite{zhang2025a,Wong2025recp}. The pixel-based designs in particular have further solidified the practical feasibility of FAS, enabling precise electromagnetic signal and information processing capabilities and high-speed port switching. Recent spatial block-correlation models can provide an alternative tractable abstraction by grouping ports into correlation blocks \cite{blockCorrelationFAS}.

The immense potential of FAS has catalyzed a surge of research efforts aiming to integrate this technology across diverse network domains. In the realm of multiple access, the concept of fluid antenna multiple access (FAMA) has been proposed to manage massive connectivity \cite{wong2022fama}. By exploiting the deep fading dips inherent in the spatial domain, FAMA allows users to avoid interference by selecting ports in which interference signals are naturally suppressed \cite{hong2025multi}. Subsequent works have extended this paradigm to enhance both orthogonal and non-orthogonal multiple access (NOMA) schemes, utilizing optimization algorithms and large language models (LLMs) to jointly optimize port selection and beamforming vectors \cite{new2024fanoma, guo2026llm}. Beyond terrestrial networks, the versatility of FAS has also been demonstrated in integrated sensing and communication (ISAC) and energy harvesting scenarios \cite{lin2025flui}. Furthermore, the principle of ``fluidity" has been generalized to reconfigurable intelligent surfaces (RIS), providing a new degree of freedom (DoF) to passive reflection \cite{xiao2025from, zhu2025fris, yao2025per, wu2025unlock}. This flexibility is particularly advantageous in dynamic three-dimensional (3D) networks involving unmanned aerial vehicles (UAVs) \cite{zhu202ffasuvg,zhu2025fasfinit} and satellite communications \cite{zhao2025del}, where channel conditions fluctuate rapidly. On the theoretical front, fundamental performance limits such as outage probability, ergodic capacity, and diversity-multiplexing tradeoffs have been rigorously characterized for various MIMO-FAS configurations, establishing a solid information-theoretic foundation for the technology \cite{new2024fas, new2024an, zhang2025onfun, zhang2025coded,zhang2025onfl,zhu25gfas}.

Despite this growing body of literature, a critical practical challenge remains largely unaddressed, particularly in the context of latency-sensitive HRLLC: the port switching delay. The vast majority of existing theoretical frameworks rely on the assumption of ideal instantaneous port switching or evaluate performance based on infinite blocklength assumptions \cite{he2024mov}. However, in realistic hardware implementations, the processes of port scanning, channel estimation, and radio-frequency (RF) chain switching consume non-negligible time resources. In the strict latency regime of 6G HRLLC, where transmission occurs over short packets governed by finite blocklength information theory, this switching overhead becomes a dominant performance-limiting factor. On one hand, a higher number of ports enhances the probability of finding a superior signal path. On the other hand, scanning more ports linearly depletes the time resources available for payload transmission.

Given a fixed total latency budget, the effective blocklength for data transmission shrinks. According to finite blocklength theory, this reduction in blocklength inevitably incurs a severe finite blocklength penalty, increasing the decoding error probability. Ignoring this temporal cost leads to a significant overestimation of system reliability and erroneous hardware dimensioning. Though sparse studies have touched upon finite blocklength in scenarios like UAV-FAS \cite{zhu2025fasfinit}, a comprehensive framework explicitly modeling the trade-off between spatial correlation, hardware switching latency, and short-packet reliability in point-to-point HRLLC is still lacking.

To bridge this gap, this paper conducts a systematic investigation into the performance of FAS-enabled 6G HRLLC systems, explicitly accounting for practical port switching delays. We establish a rigorous analytical framework rooted in finite blocklength information theory to quantify the impact of hardware overheads on system reliability and throughput. Our main contributions are summarized as follows:
\begin{itemize}
\item [1)]We propose a delay-aware FAS framework for finite blocklength HRLLC, where the effective blocklength is explicitly coupled with the number of available ports. This model captures the practical cost of port scanning, channel estimation, and switching, thereby revealing the fundamental trade-off between spatial diversity gain and signaling-latency penalty. This formulation provides a practical basis for studying FAS in latency-critical 6G user terminals.
\item [2)] We develop a tractable analytical framework for evaluating the average block error rate (BLER) and average achievable rate of spatially correlated FAS channels. To handle the analytical difficulty caused by correlated physical port gains, we use the Jakes model to construct the spatial covariance matrix and adopt an effective spatial DoF representation based on its dominant eigenmodes. Under this representation, closed-form analytical approximations are derived using the Inclusion-Exclusion Principle, while the main impact of spatial correlation is preserved through the eigenvalue spectrum.
\item [3)] We characterize the reliability-rate trade-off induced by port switching delay under finite blocklength operation. In particular, we derive switching-delay thresholds that identify the operating regimes where FAS can provide performance gains over a conventional FPA system. This analysis shows that increasing the number of ports improves spatial diversity but also shortens the effective blocklength, leading to a non-monotonic performance trend.
\item [4)] We formulate reliability-, throughput-, and energy-efficiency (EE)-oriented port-dimensioning problems over the feasible integer port set. The continuous-relaxation analysis is used to explain the structural behavior of the objective functions, while the actual discrete optimum is obtained through finite search. A tie-breaking rule is included to select the smallest optimal port number when multiple integer solutions give the same objective value. Numerical results validate the analytical framework and show that FAS can achieve substantial HRLLC gains when the switching delay is sufficiently small, whereas excessive port scanning may degrade performance due to the finite-blocklength penalty.
\end{itemize}


\textit{Notations}: Scalars, vectors, and matrices are denoted by italic, boldface lower-case, and boldface upper-case letters, respectively. The superscripts $(\cdot)^T$ and $(\cdot)^H$ denote transpose and Hermitian transpose. $\mathbb{C}^{M\times N}$ denotes the space of $M\times N$ complex matrices, and $\mathbb{E}\{\cdot\}$ denotes expectation. $|\cdot|$ denotes the absolute value or modulus. $Q(\cdot)$ and $Q^{-1}(\cdot)$ denote the Gaussian Q-function and its inverse, respectively. $J_0(\cdot)$ denotes the zero-order Bessel function of the first kind.
\vspace{-1mm}
\section{System Model}\label{sec:system_model}
\subsection{System Configuration}
Consider a point-to-point (P2P) downlink system for 6G HRLLC, where a base station (BS) serves a user terminal (UT) equipped with a FAS. To isolate the fundamental impact of UT-side port switching delay under finite blocklength operation, we consider a single-stream baseline and model the BS with one transmit antenna. This setting captures the core trade-off between FAS spatial diversity and port-scanning overhead without introducing BS beamforming, feedback, or multi-antenna training overhead. It can also be regarded as a canonical single-stream baseline, or as a special case of a multi-antenna BS system with a fixed transmit beam or one active transmit antenna. In Fig. \ref{frame}, we model the $N$ preset ports as being uniformly distributed along a linear space of normalized length $W$, where the physical length is $W\lambda$ with $\lambda$ denoting the carrier wavelength.

\vspace{-2mm}
\subsection{Spatially Correlated Channel Modeling}
\label{subsec:channel_model}
Let $\mathbf{h}=[h_1,h_2,\ldots,h_N]^T\in\mathbb{C}^{N\times1}$ denote the channel state information (CSI) vector, where $h_n$ is the channel gain at the $n$-th port. We consider a quasi-static spatially correlated Rayleigh fading channel within one HRLLC packet, where port scanning, CSI acquisition, port selection, and data transmission are assumed to occur within the channel coherence interval. Therefore, the CSI used for port selection is assumed to remain valid during the data-transmission phase. The channel is modeled as
\begin{equation}
\mathbf{h}\sim\mathcal{CN}(\mathbf{0},\mathbf{J}),
\end{equation}
where $\mathbf{J}=\mathbb{E}\{\mathbf{h}\mathbf{h}^H\}\in\mathbb{C}^{N\times N}$ is the spatial covariance matrix. Hence, each port experiences Rayleigh fading marginally, i.e., $h_n\sim\mathcal{CN}(0,\sigma^2)$, while different ports are spatially correlated.

Using the Jakes' model for isotropic scattering, the spatial correlation between the $m$-th and $n$-th ports is modeled as
\begin{equation} \label{eq:1d_jakes}
J_{m,n} = \sigma^2 J_0 \left( 2\pi \frac{|m-n|}{N-1} W \right),
\end{equation}
where $\sigma^2$ represents the large-scale fading coefficient and $J_{m,n}=\mathbb{E}\{h_m h_n^{*}\}$. In particular, $J_{n,n}=\sigma^2$, which is consistent with the marginal distribution $h_n \sim \mathcal{CN}(0,\sigma^2)$.

To facilitate performance analysis, we decompose the spatial covariance matrix as $\mathbf{J}=\mathbf{U}\mathbf{\Lambda}\mathbf{U}^H$, where $\mathbf{\Lambda}={\rm diag}(\lambda_1,\ldots,\lambda_N)$ contains the eigenvalues in descending order. The correlated channel can be generated as $\mathbf{h}=\mathbf{U}\mathbf{\Lambda}^{1/2}\mathbf{z}$, where $\mathbf{z}\sim\mathcal{CN}(\mathbf{0},\mathbf{I}_N)$ contains independent complex Gaussian entries. This decomposition characterizes the dominant spatial eigenmodes and the effective DoF induced by the Jakes covariance matrix. The post-selection SNR over the physical ports is
\begin{equation}
\gamma_{\mathrm{FAS}}=\bar{\gamma}\max_{n\in\{1,\ldots,N\}}|h_n|^2,
\end{equation}
where $\bar{\gamma}=P_{\mathrm{t}}/N_0$ is the transmit SNR.

\vspace{-2mm}
\subsection{Frame Structure and Latency Modeling}
Let $L_{\mathrm{tot}}$ denote the total latency budget in channel uses allocated for one HRLLC packet. To account for practical signaling overhead, we define $\tau$ as the aggregate overhead per port.\footnote{The overhead $\tau$ includes RF switching, port scanning, pilot signaling, channel estimation, and control signaling, rather than only pilot duration. A common $\tau$ models homogeneous ports using the same RF switching network and scanning protocol; if port-dependent overheads ${\tau_n}$ exist, $\tau$ can be interpreted as an average or conservative per-port overhead, so that $\sum_{n=1}^{N}\tau_n\approx N\tau$. Port selection is assumed to use sufficiently accurate and non-outdated CSI within the channel coherence interval; hence, the results serve as a quasi-static ideal-CSI benchmark.} Consequently, the effective blocklength available for channel coding is
\begin{equation} \label{eq:effective_L}
L(N)=L_{\mathrm{tot}}-N\tau .
\end{equation}

As shown in Fig.~\ref{frame}, each transmission block is divided into a port-selection phase and a data-transmission phase. The port-selection phase includes the practical overhead required for port probing/scanning, CSI acquisition, control signaling, and RF switching, and consumes $N\tau$ channel uses. Hence, increasing $N$ improves the chance of selecting a strong spatial path but reduces the effective blocklength available for finite-blocklength data transmission.

\begin{figure}[t]
\centering
\includegraphics[width=.75\columnwidth]{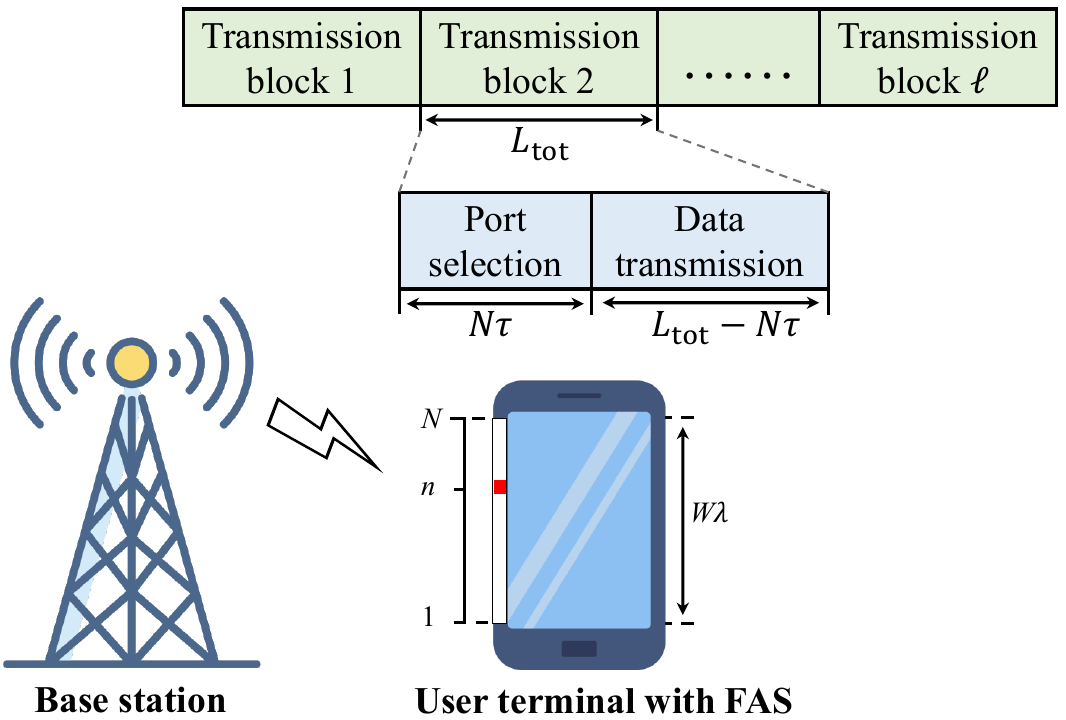}
\caption{System model and delay-aware transmission block for FAS-enabled HRLLC.}
\label{frame}
\vspace{-3mm}
\end{figure}
\subsection{Finite Blocklength Coding Model}
For a short packet containing $D$ information bits transmitted over the effective blocklength $L(N)$, the BLER is approximated using the normal approximation \cite{polyanskiy2010channel}
\begin{equation} \label{eq:bler_approx}
\epsilon(\gamma, N) \approx Q\left( \frac{C(\gamma) - \frac{D}{L(N)}}{\sqrt{V(\gamma)/L(N)}} \right),
\end{equation}
where $C(\gamma) = \log_2(1+\gamma)$ represents the Shannon capacity, and $V(\gamma) = (1-(1+\gamma)^{-2})(\log_2 e)^2$ denotes the channel dispersion.
The normal approximation is used under the condition $L(N)\geq L_{\min}$, with $L_{\min}=100$ unless otherwise specified. This framework facilitates the joint optimization of reliability and latency, allowing us to identify the optimal port number $N^{\star}$ that balances diversity benefits against the signaling overhead.

\vspace{-2mm}
\section{Performance Analysis}\label{sec:performance_analysis}
In this section, we derive a closed-form analytical approximation for the average BLER to quantify the reliability-latency trade-off.

\vspace{-2mm}
\subsection{Statistical Characterization}
Recalling that $\mathbf{J}=\mathbf{U}\mathbf{\Lambda}\mathbf{U}^H$, the exact distribution of the maximum over correlated physical port gains generally requires the joint distribution of $\{|h_n|^2\}_{n=1}^{N}$ and cannot be factorized into a product of marginal distributions. To obtain a tractable analytical form while retaining the dominant impact of spatial correlation, we adopt an effective independent spatial-DoF representation based on the dominant eigenmodes of $\mathbf{J}$. Specifically, the analytical counterpart of the post-selection SNR is expressed as
\begin{equation}
    \tilde{\gamma}_{\mathrm{FAS}}
    =
    \bar{\gamma}\max_{m\in\{1,\ldots,M\}}\lambda_m |z_m|^2,
    \label{eq:effective_snr}
\end{equation}
where $z_m\sim\mathcal{CN}(0,1)$ are mutually independent and $M$ is the effective spatial rank. Under this representation, the Inclusion-Exclusion Principle can be applied to obtain a tractable closed-form analytical approximation.

\begin{proposition} \label{prop:exact_pdf}
Under the effective independent spatial-DoF representation in \eqref{eq:effective_snr}, the probability density function (PDF) of the analytical post-selection SNR $\tilde{\gamma}_{\rm FAS}$ is approximated by a linear combination of exponential functions
\begin{equation} \label{eq:pdf_exact}
f_{\tilde{\gamma}_{\rm FAS}}(x) = \sum_{k=1}^{M} (-1)^{k+1} \sum_{\mathbf{s} \in \mathcal{S}_k} \frac{\Xi_{\mathbf{s}}}{\bar{\gamma}} \exp\left( - \frac{\Xi_{\mathbf{s}}}{\bar{\gamma}} x \right),
\end{equation}
where $\mathcal{S}_k$ represents the set of all $k$-combinations of indices $\{1,2,\ldots,M\}$. The effective spatial rank $M$ is determined by the following energy-retention criterion:
\begin{equation}
M=\min\left\{m\in\{1,\ldots,N\}:\frac{\sum_{i=1}^{m}\lambda_i}{\sum_{i=1}^{N}\lambda_i}\geq \eta\right\}.
\label{eq:effective_rank}
\end{equation}
where the eigenvalues are sorted in descending order and $\eta=0.99$ is used unless otherwise specified. This criterion ensures that at least $99\%$ of the covariance energy is retained. The residual normalized energy is given by
\begin{equation}
\rho_M=1-\frac{\sum_{i=1}^{M}\lambda_i}{\sum_{i=1}^{N}\lambda_i}\leq 1-\eta.
\label{eq:residual_energy}
\end{equation}
Thus, $M$ is selected in a reproducible way rather than by an empirical truncation rule. Fig.~\ref{Fig10} illustrates the eigenvalue spectrum of $\mathbf{J}$ and shows that the dominant spatial energy is concentrated in the first few eigenmodes due to correlation-induced DoF saturation.

For a specific combination $\mathbf{s}\in\mathcal{S}_k$, the aggregate decay rate is defined as $\Xi_{\mathbf{s}}\triangleq\sum_{j\in\mathbf{s}}\lambda_j^{-1}$.
\end{proposition}

\begin{figure}[t]
\centering
\includegraphics[width=.7\columnwidth]{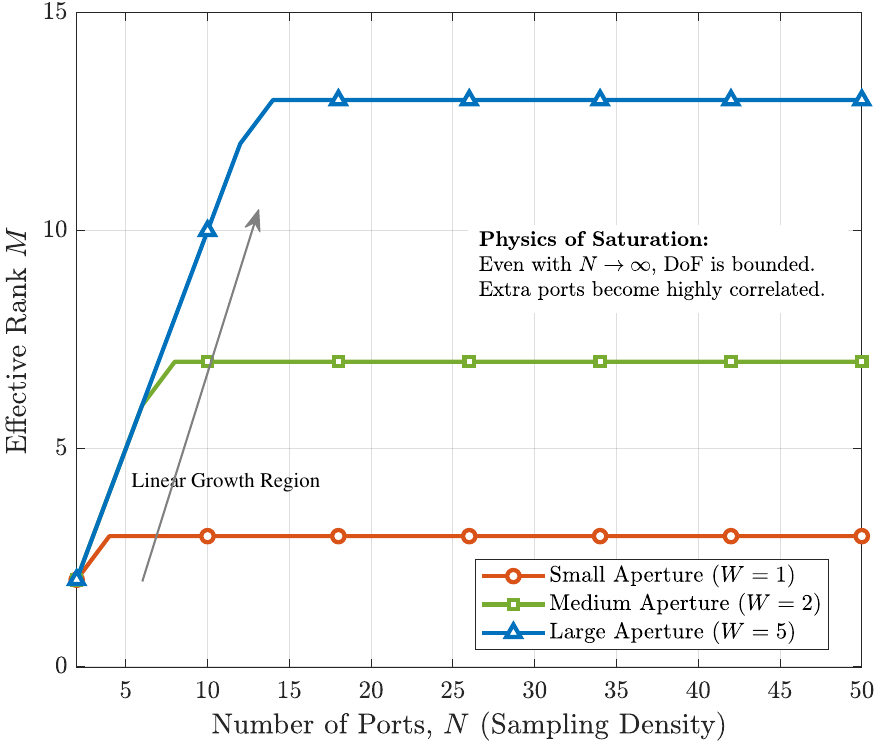}\\
\caption{Eigenvalue spectrum of the spatial correlation matrix illustrating the effective DoF $M$ used in performance analysis.}\label{Fig10}
\vspace{-3mm}
\end{figure}
\begin{IEEEproof}

Under the effective representation in \eqref{eq:effective_snr}, define
\begin{equation}
Y_m=\bar{\gamma}\lambda_m |z_m|^2,\quad m=1,\ldots,M .
\end{equation}
Since $z_m\sim\mathcal{CN}(0,1)$ are mutually independent, $Y_m$ are independent exponential random variables with
\begin{equation}
F_{Y_m}(x)=
1-\exp\left(-\frac{x}{\bar{\gamma}\lambda_m}\right),
\quad x\geq0.
\end{equation}
Thus, the cumulative distribution function (CDF) of $\tilde{\gamma}_{\rm FAS}=\max_{m}Y_m$ is
\begin{equation}
\begin{aligned}
F_{\tilde{\gamma}_{\rm FAS}}(x)
&=
\prod_{m=1}^{M}
\left(
1-\exp\left(-\frac{x}{\bar{\gamma}\lambda_m}\right)
\right) \\
&=
1+
\sum_{k=1}^{M}(-1)^k
\sum_{\mathbf{s}\in\mathcal{S}_k}
\exp\left(-\frac{\Xi_{\mathbf{s}}}{\bar{\gamma}}x\right),
\label{eq:cdf_expanded}
\end{aligned}
\end{equation}
where $\Xi_{\mathbf{s}}\triangleq\sum_{j\in\mathbf{s}}\lambda_j^{-1}$. The product form follows from the independence of the effective eigenmode variables $Y_m$, rather than from the independence of the physical port gains. The second equality follows from the inclusion-exclusion expansion. Differentiating \eqref{eq:cdf_expanded} with respect to $x$ yields \eqref{eq:pdf_exact}.
\end{IEEEproof}

The subset expansion in \eqref{eq:pdf_exact} contains \(2^M-1\) exponential terms. Therefore, for a fixed \(N\), the evaluation complexity of the PDF or CDF is \(\mathcal{O}(2^M)\) after the eigenvalues and subset rates \(\{\Xi_{\mathbf{s}}\}\) are obtained. The complexity is exponential in the effective spatial rank \(M\), but not directly in the physical port number \(N\). Since compact FAS apertures usually exhibit strong spatial correlation, \(M\) is typically much smaller than \(N\) under the energy-retention rule in \eqref{eq:effective_rank}. If all \(N\) eigenmodes are used, the full product-form CDF is denoted by \(F_N(x)\); using the first \(M\) eigenmodes gives \(F_M(x)\). Their difference satisfies
\begin{equation}
0\leq F_M(x)-F_N(x)
\leq
\sum_{i=M+1}^{N}\exp\left(-\frac{x}{\bar{\gamma}\lambda_i}\right),
\label{eq:cdf_truncation_bound}
\end{equation}
which provides a quantitative indicator of the truncation effect. Since the average BLER and achievable-rate expressions are subsequently evaluated based on the same analytical SNR distribution, this CDF-level discrepancy also indicates how eigenmode truncation affects the derived performance metrics. Hence, the truncation accuracy can be controlled by increasing $\eta$.

\begin{remark}
\textbf{(Eigenvalues and  Correlation)} The eigenvalues $\{\lambda_n\}$ serve as a spectral representation of the spatial correlation structure defined in \eqref{eq:1d_jakes}. As the port density increases within the fixed linear aperture $W\lambda$, the spatial correlation strengthens, causing smaller eigenvalues to vanish. This implies that the effective DoF $M$ are upper-bounded by the physical dimension of the FAS, limiting the diversity gain.
\end{remark}

\vspace{-2mm}
\subsection{Closed-Form Average BLER Derivation}
The average BLER is defined as the expectation of the instantaneous error probability over the fading statistics
\begin{equation} \label{eq:avg_bler_integral}
\bar{\epsilon}(N) = \int_{0}^{\infty} Q\left( \frac{C(x) - \frac{D}{L(N)}}{\sqrt{V(x)/L(N)}} \right) f_{\tilde{\gamma}_{\rm FAS}}(x) dx,
\end{equation}
which is analytically intractable due to the nonlinearity of the Gaussian Q-function. To circumvent this intractability while maintaining high accuracy in the HRLLC regime, we employ the first-order linearization approximation \cite{zhu2025fasfinit}. In this manner, the instantaneous BLER is approximated as
\begin{equation} \label{eq:linearized_Q}
\epsilon(x) \approx \begin{cases}
1, & x < \delta_L \\
\frac{1}{2} - \beta_N (x - \theta_N), & \delta_L \le x \le \delta_H \\
0, & x > \delta_H,
\end{cases},
\end{equation}
in which the linearization parameters are functions of $L(N)$ given by $\theta_N = 2^{\frac{D}{L(N)}} - 1$ and $\beta_N = \sqrt{\frac{L(N)}{2\pi (2^{2D/L(N)}-1)}}$. The integration limits are $\delta_L = \theta_N - \frac{1}{2\beta_N}$ and $\delta_H = \theta_N + \frac{1}{2\beta_N}$.

Substituting \eqref{eq:pdf_exact} and \eqref{eq:linearized_Q} into \eqref{eq:avg_bler_integral}, we obtain a closed-form analytical approximation for $\bar{\epsilon}(N)$ as follows.

\begin{theorem}\label{thm:closed_form_bler}
The average BLER of the P2P FAS under finite blocklength constraints is approximated by
\begin{multline}\label{eq:closed_form_result}
\bar{\epsilon}(N) \approx F_{{\tilde\gamma}_{\rm FAS}}(\delta_L) + \left( \frac{1}{2} + \beta_N \theta_N \right) \left[ F_{{\tilde\gamma}_{\rm FAS}}(\delta_H) - F_{{\tilde\gamma}_{\rm FAS}}(\delta_L) \right]\\
- \beta_N \sum_{k=1}^{M} (-1)^{k+1} \sum_{\mathbf{s} \in \mathcal{S}_k} \left[ \mathcal{H}(\delta_L, \Xi_{\mathbf{s}}) - \mathcal{H}(\delta_H, \Xi_{\mathbf{s}}) \right],
\end{multline}
where $\mathcal{H}(y, \Xi) = \exp\left( - \frac{\Xi}{\bar{\gamma}}y \right) \left( y + \frac{\bar{\gamma}}{\Xi} \right)$ and $F_{\tilde{\gamma}_{\rm FAS}}(\cdot)$ is the CDF derived in \eqref{eq:cdf_expanded}.
\end{theorem}

\begin{IEEEproof}
Substituting the linearized BLER approximation \eqref{eq:linearized_Q} into the definition of average BLER \eqref{eq:avg_bler_integral}, we partition the integration domain $[0, \infty)$ into three distinct regions: the outage region $[0, \delta_L)$, the transition region $[\delta_L, \delta_H]$, and the error-free region $(\delta_H, \infty)$. As such, we have
\begin{align} \label{eq:proof_split}
\bar{\epsilon}(N) &\approx \int_{0}^{\delta_L}  f_{\tilde{\gamma}_{\mathrm{FAS}}}(x) \, dx\! + \! \int_{\delta_L}^{\delta_H} \left[ \frac{1}{2} \!- \!\beta_N (x \!-\! \theta_N) \right] f_{\tilde{\gamma}_{\mathrm{FAS}}}(x) \, dx \nonumber \\
&= \underbrace{F_{\tilde{\gamma}_{\mathrm{FAS}}}(\delta_L)}_{\text{Term I}} + \underbrace{\left( \frac{1}{2} + \beta_N \theta_N \right) \int_{\delta_L}^{\delta_H} f_{\tilde{\gamma}_{\mathrm{FAS}}}(x) \, dx}_{\text{Term II}} \nonumber \\
&\quad - \underbrace{\beta_N \int_{\delta_L}^{\delta_H} x f_{\tilde{\gamma}_{\mathrm{FAS}}}(x) \, dx}_{\text{Term III}}.
\end{align}

In this step, we proceed to evaluate these terms sequentially. Term I is directly identified as the CDF evaluated at the lower limit $\delta_L$. The second term, Term II, represents the probability mass accumulated within the transition region scaled by a constant factor. This simplifies to the difference of the CDF values at the boundaries, given by
\begin{equation}
\text{ Term II} = \left( \frac{1}{2} + \beta_N \theta_N \right) \left[ F_{\tilde{\gamma}_{\mathrm{FAS}}}(\delta_H) - F_{\tilde{\gamma}_{\mathrm{FAS}}}(\delta_L) \right].
\end{equation}

The evaluation of the third term, Term III, requires integrating the product $x f_{\tilde{\gamma}_{\mathrm{FAS}}}(x)$. By substituting the analytical PDF approximation in \eqref{eq:pdf_exact} into the integral, we obtain
\begin{equation}
\int_{\delta_L}^{\delta_H} x f_{\tilde{\gamma}_{\mathrm{FAS}}}(x) dx = \sum_{k=1}^{M} (-1)^{k+1} \sum_{\mathbf{s} \in \mathcal{S}_k} \int_{\delta_L}^{\delta_H} x \left( \frac{\Xi_{\mathbf{s}}}{\bar{\gamma}} e^{-\frac{\Xi_{\mathbf{s}}}{\bar{\gamma}} x} \right) dx.
\end{equation}
Leveraging the standard integration identity
\begin{equation}
\int t a e^{-a t} dt = -e^{-a t}\left(t + \frac{1}{a}\right),
\end{equation}
and defining $a = \Xi_{\mathbf{s}}/\bar{\gamma}$, the definite integral for each exponential component is calculated as
\begin{equation}
\left[ -e^{-\frac{\Xi_{\mathbf{s}}}{\bar{\gamma}} x} \left( x + \frac{\bar{\gamma}}{\Xi_{\mathbf{s}}} \right) \right]_{\delta_L}^{\delta_H} = \mathcal{H}(\delta_L, \Xi_{\mathbf{s}}) - \mathcal{H}(\delta_H, \Xi_{\mathbf{s}}),
\end{equation}
where we have defined $\mathcal{H}(y, \Xi) = \exp\left( - \frac{\Xi}{\bar{\gamma}}y \right) \left( y + \frac{\bar{\gamma}}{\Xi} \right)$.

Finally, substituting the evaluated Term II and Term III back into \eqref{eq:proof_split} yields the closed-form approximation in \eqref{eq:closed_form_result}.
\end{IEEEproof}

\vspace{-2mm}
\subsection{Asymptotic Analysis and Diversity Order}\label{subsec:asymptotic}
To quantify the diversity benefits and finite blocklength penalties, we analyze the average BLER in the high-SNR regime, i.e., $\bar{\gamma} \to \infty$. The diversity order is defined as
\begin{equation}
G_d = - \lim_{\bar{\gamma} \to \infty} \frac{\log \bar{\epsilon}(N)}{\log \bar{\gamma}}.
\end{equation}

Applying the first-order Taylor expansion $1-e^{-y}\approx y$ to the effective-mode CDF in \eqref{eq:cdf_expanded}, the individual product terms simplify to
\begin{equation}
1 - \exp\left( - \frac{x}{\bar{\gamma} \lambda_n} \right) \approx \frac{x}{\bar{\gamma} \lambda_n}.
\end{equation}

Defining the \textit{array gain} as $G_a = \left(\prod\limits_{n=1}^{M} \lambda_n\right)^{\frac{1}{M}}$, the asymptotic CDF is rewritten as
\begin{equation} \label{eq:cdf_asymp}
F_{\tilde{\gamma}_{\rm FAS}}^{\infty}(x) \approx \prod_{n=1}^{M} \left( \frac{x}{\bar{\gamma} \lambda_n} \right) = \left( \frac{x}{G_a \bar{\gamma}} \right)^M.
\end{equation}

Substituting the corresponding asymptotic PDF into the BLER integral yields the following corollary.

\begin{corollary} \label{cor:asymptotic}
As $\bar{\gamma} \to \infty$, the average BLER decays asymptotically as
\begin{equation} \label{eq:bler_asymptotic}
\bar{\epsilon}^{\infty}(N) \approx \mathcal{K}(M, L(N)) \cdot \left( G_a \cdot \bar{\gamma} \right)^{-M},
\end{equation}
where $M$ is the effective spatial rank determined by \eqref{eq:effective_rank}, which characterizes the effective diversity order under the adopted eigenmode representation, and $\mathcal{K}(M, L(N))$ is the coding gain penalty coefficient given by
\begin{equation}
\mathcal{K}(M, L(N)) = \int_{0}^{\infty} Q\left( \frac{C(x) - \frac{D}{L(N)}}{\sqrt{V(x)/L(N)}} \right) \frac{d}{dx}(x^M) \, dx.
\end{equation}
\end{corollary}

\begin{remark}
\textbf{(Diversity vs. Latency Penalty)} Equation \eqref{eq:bler_asymptotic} explicitly decouples the impact of port number $N$ into two competing factors:
\begin{enumerate}
\item \textbf{Diversity Gain ($M$):} Increasing $N$ enhances the rank $M$ (up to the spatial correlation limit), steepening the BLER decay slope.
\item \textbf{Coding Penalty ($\mathcal{K}$):} Increasing $N$ reduces the effective blocklength $L(N)$, increasing $\mathcal{K}$ and shifting the BLER curve horizontally (SNR penalty).
\end{enumerate}
This confirms that while FAS improves reliability via diversity, latency overhead imposes a fundamental limit, necessitating the optimization of $N^{\star}$.
\end{remark}

\vspace{-2mm}
\subsection{Reliability-Based Operational Region: FAS vs. FPA}
To provide a reliability benchmark for the proposed port-dimensioning strategy, we characterize the operating condition under which the $N$-port FAS achieves a lower asymptotic BLER than a conventional single-antenna FPA $(N=1,\tau=0)$.

For a standard FPA, setting $N=1$, $M=1$, and $G_a=1$ in \eqref{eq:bler_asymptotic} gives
\begin{equation} \label{eq:bler_inf_fpa}
\bar{\epsilon}_{\mathrm{FPA}}^{\infty}
\approx
\mathcal{K}(1,L_{\mathrm{tot}})\bar{\gamma}^{-1}.
\end{equation}

\begin{proposition} \label{prop:bler_boundary}
Under the high-SNR approximation and the local log-log sensitivity model of the coding penalty, the $N$-port FAS provides superior or equivalent reliability compared with the FPA, i.e., $\bar{\epsilon}_{\mathrm{FAS}}^{\infty}(N)\leq\bar{\epsilon}_{\mathrm{FPA}}^{\infty}$, when the unit switching delay approximately satisfies
\begin{equation} \label{eq:tau_bler_ineq}
\tau
\leq
\frac{L_{\mathrm{tot}}}{N}
\left[
1-
\left(
\frac{
G_a^{M}\bar{\gamma}^{M-1}
\mathcal{K}(1,L_{\mathrm{tot}})
}
{
\mathcal{K}(M,L_{\mathrm{tot}})
}
\right)^{-1/k}
\right]
\triangleq
\tau_{\mathrm{eq}}^{\epsilon},
\end{equation}
where $k>0$ denotes the local blocklength-sensitivity exponent of the coding-penalty coefficient.
\end{proposition}

\begin{IEEEproof}
The condition for FAS to outperform FPA in reliability is
$\bar{\epsilon}_{\mathrm{FAS}}^{\infty}(N)/
\bar{\epsilon}_{\mathrm{FPA}}^{\infty}\leq1$.
Substituting \eqref{eq:bler_asymptotic} and \eqref{eq:bler_inf_fpa} into this ratio yields
\begin{equation} \label{eq:bler_ratio_substituted}
\frac{
\mathcal{K}(M,L(N))G_a^{-M}\bar{\gamma}^{-M}
}
{
\mathcal{K}(1,L_{\mathrm{tot}})\bar{\gamma}^{-1}
}
\leq 1 .
\end{equation}
Equivalently,
\begin{equation} \label{eq:bler_ratio_simplified}
\frac{\mathcal{K}(M,L(N))}
{\mathcal{K}(1,L_{\mathrm{tot}})}
\leq
G_a^{M}\bar{\gamma}^{M-1}.
\end{equation}
To characterize the blocklength sensitivity of the coding penalty, we define the local log-log sensitivity exponent around $L_{\mathrm{tot}}$ as
\begin{equation}
k
\triangleq
-\left.
\frac{\partial \ln \mathcal{K}(M,L)}
{\partial \ln L}
\right|_{L=L_{\mathrm{tot}}}.
\label{eq:local_sensitivity_k}
\end{equation}
A first-order Taylor expansion of $\ln \mathcal{K}(M,L)$ with respect to $\ln L$ around $L_{\mathrm{tot}}$ gives
\begin{equation}
\ln
\frac{\mathcal{K}(M,L(N))}
{\mathcal{K}(M,L_{\mathrm{tot}})}
\approx
-k
\ln
\frac{L(N)}
{L_{\mathrm{tot}}},
\end{equation}
and therefore
\begin{equation}
\frac{\mathcal{K}(M,L(N))}
{\mathcal{K}(M,L_{\mathrm{tot}})}
\approx
\left(
\frac{L_{\mathrm{tot}}-N\tau}
{L_{\mathrm{tot}}}
\right)^{-k}.
\label{eq:penalty_ratio_model}
\end{equation}
This approximation should be interpreted as a local analytical model around $L_{\mathrm{tot}}$, rather than as a universal power-law identity over all blocklengths.
Substituting \eqref{eq:penalty_ratio_model} into \eqref{eq:bler_ratio_simplified} gives
\begin{equation}
\left(
\frac{L_{\mathrm{tot}}-N\tau}
{L_{\mathrm{tot}}}
\right)^{-k}
\leq
\frac{
G_a^{M}\bar{\gamma}^{M-1}
\mathcal{K}(1,L_{\mathrm{tot}})
}
{
\mathcal{K}(M,L_{\mathrm{tot}})
}.
\end{equation}
Solving the above inequality for $\tau$ yields \eqref{eq:tau_bler_ineq}.
\end{IEEEproof}

\begin{remark}
\textbf{(Reliability vs. Switching Delay)} The threshold $\tau_{\mathrm{eq}}^{\epsilon}$ quantifies when the diversity gain of FAS can compensate for the finite-blocklength penalty caused by port scanning. For $M>1$, increasing $\bar{\gamma}$ enlarges the factor $G_a^{M}\bar{\gamma}^{M-1}$, thereby increasing the allowable switching delay. In contrast, a smaller effective blocklength increases the coding penalty $\mathcal{K}(M,L(N))$, which reduces the feasible delay region.
\end{remark}

\vspace{-2mm}
\section{Average Achievable Rate Analysis}\label{sec:achievable_rate}
Here, we derive the average achievable rate to reveal how the spatial diversity gain compensates the rate loss induced by latency overhead.

\vspace{-2mm}
\subsection{Finite Blocklength Rate Model}
According to finite blocklength information theory, the maximal achievable rate $R^*$, measured in bits per channel use (bpcu), for a given blocklength $L$, error probability $\epsilon$, and instantaneous SNR $\gamma$ is approximated by
\begin{equation} \label{eq:rate_approximation}
R^*(\gamma, L, \epsilon) \approx C(\gamma) - \sqrt{\frac{V(\gamma)}{L}} Q^{-1}(\epsilon).
\end{equation}
Here, bpcu is used throughout the paper as the rate unit. Under the normalized-bandwidth convention, it is numerically equivalent to bps/Hz; however, all achievable-rate figures are labeled in bpcu for consistency.

Substituting the effective blocklength $L(N) = L_{\mathrm{tot}} - N\tau$, the average achievable rate of the FAS-assisted system is derived by taking the expectation over the fading channel statistics given by
\begin{equation} \label{eq:avg_rate_def}
\bar{R}(N) = \mathbb{E}_{\tilde{\gamma}_{\mathrm{FAS}}} \left[ C(\tilde{\gamma}_{\mathrm{FAS}}) - \sqrt{\frac{V(\tilde{\gamma}_{\mathrm{FAS}})}{L(N)}} Q^{-1}(\epsilon) \right].
\end{equation}

To facilitate tractable system optimization, we adopt the widely accepted high-SNR approximation for the dispersion term where $V(\gamma) \approx (\log_2 e)^2$. Consequently, the average rate decouples into a Shannon capacity term and a finite blocklength penalty term expressed as
\begin{equation} \label{eq:rate_decoupled}
\bar{R}(N) \approx \underbrace{\mathbb{E}[C(\tilde{\gamma}_{\mathrm{FAS}})]}_{\mathcal{R}_{\text{Shannon}}} - \underbrace{\frac{Q^{-1}(\epsilon) \log_2 e}{\sqrt{L_{\mathrm{tot}} - N\tau}}}_{\mathcal{P}_{\text{Penalty}}(N)}.
\end{equation}

\vspace{-2mm}
\subsection{Closed-Form Approximation of Achievable Rate}
The penalty term $\mathcal{P}_{\text{Penalty}}(N)$ in \eqref{eq:rate_decoupled} is deterministic and depends only on the number of ports $N$. Therefore, after applying the high-SNR dispersion approximation and the effective spatial-DoF representation, the remaining task is to obtain a closed-form analytical approximation for the Shannon capacity term $\mathcal{R}_{\text{Shannon}}$.

\begin{theorem} \label{thm:achievable_rate}
The average achievable rate of the spatially correlated P2P FAS under the finite blocklength constraint $L(N)$ is approximated by
\begin{multline} \label{eq:rate_closed_form}
\bar{R}(N) \approx \frac{1}{\ln 2} \sum_{k=1}^{M} (-1)^{k+1} \sum_{\mathbf{s} \in \mathcal{S}_k} e^{\frac{\Xi_{\mathbf{s}}}{\bar{\gamma}}} E_1\left( \frac{\Xi_{\mathbf{s}}}{\bar{\gamma}} \right)\\
- \frac{Q^{-1}(\epsilon) \log_2 e}{\sqrt{L_{\mathrm{tot}} - N\tau}},
\end{multline}
where $E_1(x)=\int_x^{\infty}t^{-1}e^{-t}dt$ denotes the exponential integral function, $M$ is the effective spatial rank determined by \eqref{eq:effective_rank}, and the aggregate decay rate is defined as $\Xi_{\mathbf{s}}=\sum_{j\in\mathbf{s}}\lambda_j^{-1}$.
\end{theorem}

\begin{IEEEproof}
The average achievable rate is decomposed into the ergodic capacity component, $\mathcal{R}_{\mathrm{Shannon}}$, and the finite blocklength penalty term. The capacity term is defined as
\begin{equation}
\mathcal{R}_{\mathrm{Shannon}} = \int_{0}^{\infty} \log_2(1+x) f_{\tilde{\gamma}_{\mathrm{FAS}}}(x) dx.
\end{equation}
Substituting the analytical PDF approximation from Proposition~\ref{prop:exact_pdf} into the integral yields
\begin{equation}
\mathcal{R}_{\mathrm{Shannon}} = \int_{0}^{\infty} \log_2(1+x) \left[ \sum_{k=1}^{M} (-1)^{k+1} \sum_{\mathbf{s} \in \mathcal{S}_k} \eta_{\mathbf{s}} e^{-\eta_{\mathbf{s}} x} \right] dx,
\end{equation}
where $\eta_{\mathbf{s}} \triangleq \Xi_{\mathbf{s}} / \bar{\gamma}$. Leveraging the linearity of the integration operator, we interchange the order of summation and integration. The problem reduces to evaluating the component integral $\mathcal{I}(\eta)$ for a generic decay rate $\eta$, defined as
\begin{equation} \label{eq:component_integral}
\mathcal{I}(\eta) \triangleq \int_{0}^{\infty} \log_2(1+x) \cdot \eta e^{-\eta x} dx.
\end{equation}
Letting $t = 1+x$, we apply the change of variables. The integral is then evaluated as
\begin{align} \label{eq:int_steps}
\mathcal{I}(\eta) &\hspace{.5mm}= \frac{\eta e^{\eta}}{\ln 2} \int_{1}^{\infty} \ln(t) e^{-\eta t} dt \nonumber \\
&{=} \frac{e^{\eta}}{\ln 2} \int_{\eta}^{\infty} \frac{e^{-u}}{u} du = \frac{e^{\eta}}{\ln 2} E_1(\eta),
\end{align}

Finally, substituting $\eta = \eta_{\mathbf{s}}$ back into the summation and subtracting the penalty term determined by the effective blocklength $L(N)$ yields the result in \eqref{eq:rate_closed_form}.
\end{IEEEproof}

\vspace{-2mm}
\subsection{Asymptotic Analysis and Insights}
To evaluate the high-SNR behavior where $\bar{\gamma} \to \infty$, we analyze the terms for small arguments $x = \Xi_{\mathbf{s}} / \bar{\gamma} \to 0$. According to the standard series expansion of the exponential integral function \cite[Eq. 5.1.11]{abramowitz1964handbook}, we have $E_1(x) = -\gamma_{\mathrm{EM}} - \ln x - \sum_{n=1}^\infty \frac{(-x)^n}{n \cdot n!}$. For $x \to 0$, the higher-order terms vanish, yielding the leading-order approximation $E_1(x) \approx -\gamma_{\mathrm{EM}} - \ln x$. In the same regime, $e^x=1+\mathcal{O}(x)$. Hence, the leading high-SNR approximation is $e^x E_1(x) \approx \ln(1/x) - \gamma_{\mathrm{EM}}$, where $\gamma_{\mathrm{EM}} \approx 0.5772$ is the Euler-Mascheroni constant. Substituting this expansion into the closed-form rate approximation, we derive the following corollary.

\begin{corollary} \label{cor:high_snr_rate}
At high SNR, the achievable rate scales as
\begin{equation}\label{eq29}
\bar{R}^{\infty}(N) \approx \log_2(\bar{\gamma}) + \mathcal{S}_{\mathrm{div}}(N) - \frac{Q^{-1}(\epsilon) \log_2 e}{\sqrt{L_{\mathrm{tot}} - N\tau}},
\end{equation}
where $\mathcal{S}_{\mathrm{div}}(N) =  - \frac{\gamma_{\mathrm{EM}}}{\ln 2} - \frac{1}{\ln 2} \sum_{k=1}^{M} (-1)^{k+1} \sum_{\mathbf{s} \in \mathcal{S}_k} \ln(\Xi_{\mathbf{s}})$ represents the spatial diversity gain offset.
\end{corollary}

\begin{figure}
\centering
\includegraphics[width=.7\columnwidth]{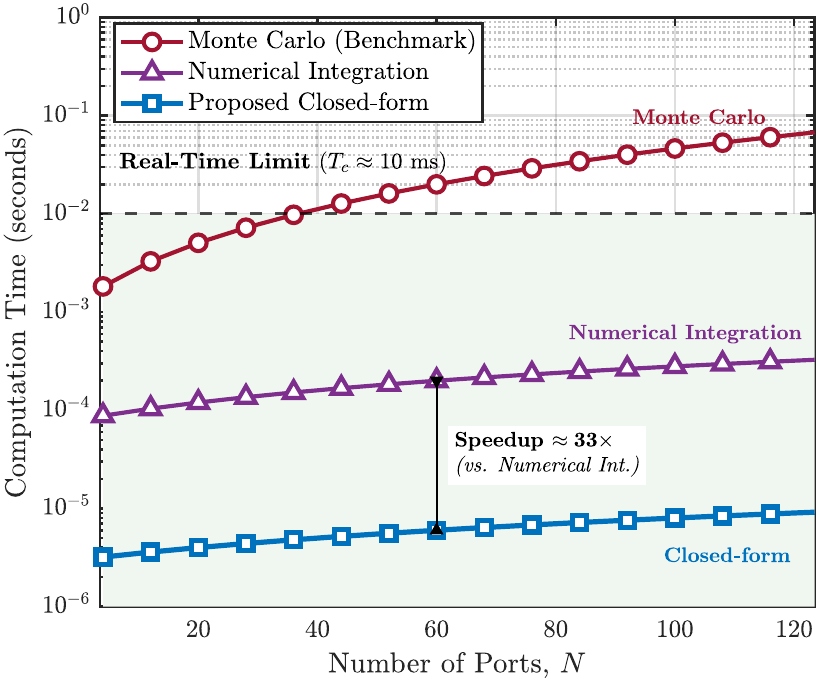}\\
\caption{Computation time versus the number of ports $N$ for different evaluation methods under the effective-rank criterion $\eta=0.99$.}\label{Fig8}
\vspace{-3mm}
\end{figure}

\begin{remark}[{\bf Finite-Blocklength Trade-off and High-SNR Validity}]
Theorem~\ref{thm:achievable_rate} and Corollary~\ref{cor:high_snr_rate} show that the achievable rate is governed by two competing effects: the spatial-diversity offset $\mathcal{S}_{\mathrm{div}}(N)$ improves the Shannon term, whereas the finite-blocklength penalty increases as the effective blocklength $L(N)=L_{\mathrm{tot}}-N\tau$ decreases. Hence, the throughput-optimal port number must be selected over the feasible integer set rather than inferred from monotonicity. The high-SNR approximations used in Corollary~\ref{cor:high_snr_rate} are intended for the moderate-to-high SNR regime relevant to HRLLC operation. In the considered numerical settings, achieving a target BLER around $10^{-5}$ typically requires an average transmit SNR of approximately $10$--$20$ dB, depending on $N$, $W$, and $\tau$. These asymptotic expressions are used only for structural insight, while the numerical results are evaluated using the finite-blocklength performance expressions.
\end{remark}

\vspace{-2mm}
\subsection{Computational Complexity Analysis}
Fig.~\ref{Fig8} compares the computation times of Monte Carlo simulation, direct numerical integration, and the proposed closed-form analytical evaluation. The tests were conducted on the same computer with an \texttt{Intel Core i7-12700H CPU}, \texttt{16 GB RAM}, running \texttt{MATLAB R2024b} on \texttt{Windows 11}. The Monte Carlo method used $S=10^6$ channel realizations, and the numerical integration used adaptive quadrature with absolute and relative tolerances of $10^{-10}$ and $10^{-8}$, respectively. The effective rank $M$ was determined by \eqref{eq:effective_rank} with $\eta=0.99$. Under this setting, the proposed method achieves up to approximately $33\times$ speedup over numerical integration. This speedup should be interpreted under the specified hardware, software, sample size, integration tolerance, and effective-rank criterion. It depends on $M$, since the subset expansion in \eqref{eq:pdf_exact} contains $2^M-1$ terms; thus, the gain decreases as $M$ increases.

\vspace{-2mm}
\subsection{Operational Region Analysis: FAS vs. FPA}
To provide a rate-oriented benchmark for the proposed port-dimensioning strategy, we characterize the operational region where the FAS provides an achievable-rate gain over a conventional FPA. The FPA is regarded as a degenerate FAS case with a single port $(N=1)$ and zero switching overhead $(\tau=0)$.

Using the high-SNR approximation in \eqref{eq29}, the average achievable rate of an $N$-port FAS is expressed as
\begin{equation} \label{eq:R_FAS_inf}
\bar{R}_{\mathrm{FAS}}^{\infty}(N)
\approx
\log_2\bar{\gamma}
+
\mathcal{S}_{\mathrm{div}}(N)
-
\frac{\chi}{\sqrt{L_{\mathrm{tot}}-N\tau}},
\end{equation}
where $\chi=Q^{-1}(\epsilon)\log_2 e$. Correspondingly, the rate of the FPA is
\begin{equation} \label{eq:R_FPA_inf}
\bar{R}_{\mathrm{FPA}}^{\infty}
\approx
\log_2\bar{\gamma}
-
\frac{\chi}{\sqrt{L_{\mathrm{tot}}}},
\end{equation}
since $\mathcal{S}_{\mathrm{div}}(1)=0$. Therefore, the rate-oriented comparison is written directly in terms of $\mathcal{S}_{\mathrm{div}}(N)$ without introducing the auxiliary notation $\Delta\mathcal{S}(N)$.

\begin{proposition} \label{prop:rate_operational_region}
Under the high-SNR rate approximation, the $N$-port FAS achieves a higher or equivalent average achievable rate compared with the FPA, i.e., $\bar{R}_{\mathrm{FAS}}^{\infty}(N)\geq\bar{R}_{\mathrm{FPA}}^{\infty}$, when the unit switching delay satisfies
\begin{equation} \label{eq:tau_rate_inequality}
\tau
\leq
\frac{L_{\mathrm{tot}}}{N}
\left[
1-
\left(
1+
\frac{\mathcal{S}_{\mathrm{div}}(N)\sqrt{L_{\mathrm{tot}}}}
{\chi}
\right)^{-2}
\right]
\triangleq
\tau_{\mathrm{eq}}^{R}.
\end{equation}
\end{proposition}

\begin{IEEEproof}
The condition for FAS to be superior or equivalent to FPA is
$\bar{R}_{\mathrm{FAS}}^{\infty}(N)-\bar{R}_{\mathrm{FPA}}^{\infty}\geq0$.
By substituting \eqref{eq:R_FAS_inf} and \eqref{eq:R_FPA_inf}, we obtain
\begin{equation} \label{eq:rate_ineq_derivation_2}
\frac{1}{\sqrt{L_{\mathrm{tot}}-N\tau}}
\leq
\frac{\mathcal{S}_{\mathrm{div}}(N)}{\chi}
+
\frac{1}{\sqrt{L_{\mathrm{tot}}}}.
\end{equation}
Multiplying both sides of \eqref{eq:rate_ineq_derivation_2} by $\sqrt{L_{\mathrm{tot}}}$ and defining
\begin{equation}
\Psi
\triangleq
1+
\frac{\mathcal{S}_{\mathrm{div}}(N)\sqrt{L_{\mathrm{tot}}}}{\chi}.
\end{equation}
Then, we have
\begin{equation}
\frac{\sqrt{L_{\mathrm{tot}}}}
{\sqrt{L_{\mathrm{tot}}-N\tau}}
\leq
\Psi.
\end{equation}
Thus, $L_{\mathrm{tot}}/(L_{\mathrm{tot}}-N\tau)\leq\Psi^2$, which implies
\begin{equation}
L_{\mathrm{tot}}-N\tau
\geq
\frac{L_{\mathrm{tot}}}{\Psi^2}.
\end{equation}
Solving the above inequality for $\tau$ and substituting the definition of $\Psi$ yields \eqref{eq:tau_rate_inequality}.
\end{IEEEproof}

\begin{remark} \label{rem:physical_insights}
\textbf{(Physical Insights)} The inequality \eqref{eq:tau_rate_inequality} delineates the feasible hardware-latency region for rate-oriented FAS operation. Since $\mathcal{S}_{\mathrm{div}}(N)$ is upper-bounded by the finite aperture-induced DoF, the allowable switching delay is constrained by the physical terminal size. Moreover, stringent HRLLC reliability requirements increase $Q^{-1}(\epsilon)$ and thus enlarge $\chi$, narrowing the feasible delay region.
\end{remark}

\vspace{-2mm}
\section{Optimal and Theoretical Properties}\label{sec:optimization}

\subsection{Problem Formulation}
We formulate the port dimensioning problem from three complementary perspectives: reliability maximization for safety-critical applications, throughput maximization for data-intensive services, and EE maximization for sustainable 6G operations. Since the number of ports $N$ is discrete and bounded, these are integer optimization problems.

\subsubsection{Reliability-Oriented Design (P1)}
For mission-critical scenarios where minimizing packet loss is paramount, the objective is to minimize the average BLER. That is,
\begin{subequations} \label{eq:opt_bler}
\begin{align}
\mathbf{P1}:  \min_{N \in \mathbb{Z}^{+}} &~ \bar{\epsilon}(N) \\
\text{s.t.} &~ 1 \le N \le \left\lfloor \frac{L_{\mathrm{tot}} - L_{\min}}{\tau} \right\rfloor, \label{eq:const_bler}
\end{align}
\end{subequations}
where \eqref{eq:const_bler} ensures that the effective blocklength $L(N)=L_{\mathrm{tot}}-N\tau$ satisfies the validity condition of the normal approximation. Unless otherwise specified, we set $L_{\min}=100$, so that all feasible port numbers obey $L(N)\geq100$.

\subsubsection{Throughput-Oriented Design (P2)}
When requiring high spectral efficiency under latency constraints, the goal is to maximize the average achievable rate,\footnote{Strictly speaking, the effective throughput is defined as $T = \bar{R}(1-\bar{\epsilon})$. However, in the context of HRLLC, the target error probability is extremely low (e.g., $\bar{\epsilon} \le 10^{-5}$), implying $1-\bar{\epsilon} \approx 1$. Therefore, the average achievable rate $\bar{R}$ serves as a precise proxy for the system throughput.} expressed as
\begin{subequations} \label{eq:opt_rate}
\begin{align}
\mathbf{P2}: \max_{N \in \mathbb{Z}^{+}} &~\bar{R}(N) \\
\text{s.t.} &~ 1 \le N \le \left\lfloor \frac{L_{\mathrm{tot}} - L_{\min}}{\tau} \right\rfloor.
\end{align}
\end{subequations}

\subsubsection{EE-Oriented Design (P3)}
Besides reliability and system throughput, EE stands as a pivotal metric for sustainable 6G networks. The deployment of FAS introduces a unique energy trade-off: while increasing $N$ enhances the achievable rate (reducing the energy-per-bit), it concurrently prolongs the port scanning duration, thereby increasing the energy consumption associated with channel estimation and RF switching.

To formulate the EE maximization problem, we first model the average power consumption based on the frame structure defined in Fig. \ref{frame}. The total energy consumed per transmission block, denoted by $E_{\mathrm{tot}}(N)$, comprises the energy dissipated during the port-selection phase and the data-transmission phase
\begin{equation} \label{eq:energy_model}
	E_{\mathrm{tot}}(N) = \underbrace{(N\tau) P_{\mathrm{scan}}}_{\text{Selection Energy}} + \underbrace{(L_{\mathrm{tot}} - N\tau) P_{\mathrm{active}}}_{\text{Transmission Energy}},
\end{equation}
where $P_{\mathrm{scan}}$ represents the circuit power consumption during the port-selection phase, including port probing/scanning, CSI acquisition, and RF switching, and $P_{\mathrm{active}}$ denotes the total power consumption during data transmission (including transmit power and RF chain processing).
The average power consumption is then given by $P_{\mathrm{avg}}(N) = E_{\mathrm{tot}}(N) / L_{\mathrm{tot}}$.

Consequently, the EE metric, defined as the ratio of the average system throughput to the average power consumption (in bits/Joule), is expressed as
\begin{equation} \label{eq:ee_metric}
	\eta_{\mathrm{EE}}(N) = \frac{\overline{R}(N)}{P_{\mathrm{avg}}(N)} = \frac{L_{\mathrm{tot}} \cdot \overline{R}(N)}{N\tau P_{\mathrm{scan}} + (L_{\mathrm{tot}} - N\tau) P_{\mathrm{active}}}.
\end{equation}

The EE maximization problem is thus formulated as
\begin{subequations}
\begin{align}
\textbf{P3:} \max_{N \in \mathbb{Z}^{+}} &~\eta_{\mathrm{EE}}(N) \\
		\text{s.t.} &~ 1 \le N \le \left\lfloor \frac{L_{\mathrm{tot}} - L_{\min}}{\tau} \right\rfloor, \label{eq:c1_ee} \\
		&~ \overline{\epsilon}(N) \le \epsilon_{\mathrm{th}}, \label{eq:c2_ee}
\end{align}
\end{subequations}
where \eqref{eq:c2_ee} imposes a reliability constraint ensuring that the pursuit of EE does not compromise the fundamental HRLLC requirement, with $\epsilon_{\mathrm{th}}$ being the maximum tolerable BLER.

\begin{remark}
\textbf{(EE Trend)} The EE objective function $\eta_{\mathrm{EE}}(N)$ typically exhibits a non-monotonic trend with respect to $N$. The numerator $\bar{R}(N)$ benefits from spatial diversity but is limited by DoF saturation and finite-blocklength penalty, while the denominator $P_{\mathrm{avg}}(N)$ increases with the port-scanning overhead. This interaction explains the observed EE trade-off. The EE-optimal port number is obtained by finite search over the feasible integer set rather than by assuming a unique optimizer from the continuous relaxation.
\end{remark}

\vspace{-2mm}

\subsection{Unimodality Analysis} \label{subsec}
The derived analytical expressions provide a quantitative measure of system performance. To gain design insights, we analyze the structural properties of the reliability, throughput, and EE metrics by relaxing the discrete port number $N$ to a continuous variable $x\in[1,N_{\max}]$. It should be emphasized that this continuous-relaxation analysis is used to explain the non-monotonic behavior of the objective functions and to guide port-dimensioning design. The actual optimization problems $\mathbf{P1}$-$\mathbf{P3}$ are still solved over the finite integer feasible set. Therefore, the discrete global optimum is obtained by direct enumeration rather than by directly rounding the continuous optimizer.

\subsubsection{High-SNR Reliability Trend}
We first investigate the reliability-oriented metric. Let
$\mathcal{L}(x)=\ln\bar{\epsilon}(x)$ denote the continuous relaxation of the log-BLER function, where $x\in[1,N_{\max}]$ is the relaxed port number.
\begin{proposition} \label{prop:convexity_bler}
Under the high-SNR asymptotic approximation in Corollary~\ref{cor:asymptotic}, the relaxed log-BLER admits the following decomposition:
\begin{equation}
\mathcal{L}^{\infty}(x)
\approx
\ln \mathcal{K}(M(x),L(x))
-
M(x)\ln\left(G_a(x)\bar{\gamma}\right),
\label{eq:bler_log_decomp}
\end{equation}
where $L(x)=L_{\mathrm{tot}}-x\tau$. This decomposition provides an asymptotic explanation for the U-shaped reliability behavior with respect to the port number.
\end{proposition}

\begin{IEEEproof}
From Corollary~\ref{cor:asymptotic}, the high-SNR average BLER is approximated as
\begin{equation}
\bar{\epsilon}^{\infty}(x)
\approx
\mathcal{K}(M(x),L(x))
\left(G_a(x)\bar{\gamma}\right)^{-M(x)}.
\end{equation}
Taking the logarithm of both sides yields
\begin{equation}
\ln\bar{\epsilon}^{\infty}(x)
\approx
\ln \mathcal{K}(M(x),L(x))-
M(x)\ln\left(G_a(x)\bar{\gamma}\right),
\end{equation}
which gives \eqref{eq:bler_log_decomp}. The first term represents the finite-blocklength coding penalty induced by the reduced effective blocklength $L(x)$, while the second term represents the reliability gain brought by the effective spatial diversity. As $x$ increases, the diversity gain gradually saturates due to the finite spatial DoF of the compact aperture, whereas the coding penalty increases because $L(x)=L_{\mathrm{tot}}-x\tau$ decreases linearly. This explains the non-monotonic reliability trend observed in the numerical results.

It should be emphasized that \eqref{eq:bler_log_decomp} is derived from the high-SNR asymptotic approximation and is not used as a finite-SNR convexity proof. The finite-SNR discrete optimum is obtained by direct enumeration over the feasible integer set in Algorithm~\ref{alg:optimization}.
\end{IEEEproof}

\subsubsection{High-SNR Throughput Trend}\label{subsec:high_snr_throughput_trend}
We next investigate the throughput-oriented metric. Corollary~\ref{cor:high_snr_rate} shows that, in the high-SNR regime, the average achievable rate can be decomposed into an SNR-dependent baseline term, a spatial-diversity offset, and a finite-blocklength penalty. By relaxing the integer port number $N$ to a continuous variable $x\in[1,N_{\max}]$, this trend can be written as
\begin{equation}
\bar{R}^{\infty}(x)
\approx
\log_2(\bar{\gamma})
+
\mathcal{S}_{\mathrm{div}}(x)
-
\frac{Q^{-1}(\epsilon)\log_2 e}
{\sqrt{L_{\mathrm{tot}}-x\tau}} .
\label{eq:rate_high_snr_decomp}
\end{equation}
The first term is independent of $x$, the second term represents the spatial-diversity offset induced by the effective spatial DoF, and the last term represents the finite-blocklength penalty caused by the reduced effective blocklength. As $x$ increases, the diversity-related gain gradually saturates because the compact aperture provides only a finite number of effective spatial DoF. Meanwhile, the finite-blocklength penalty increases because $L_{\mathrm{tot}}-x\tau$ decreases with $x$. The competition between these two effects explains the non-monotonic throughput trend observed in the numerical results. It should be emphasized that \eqref{eq:rate_high_snr_decomp} is derived from the high-SNR asymptotic approximation and is not used as a finite-SNR concavity proof. The finite-SNR discrete throughput optimum is obtained by direct enumeration over the feasible integer set in Algorithm~\ref{alg:optimization}.

\begin{remark}
The continuous-relaxation properties above provide useful structural insights into the port-dimensioning problem. They explain why increasing the number of ports first improves performance through spatial diversity but eventually becomes detrimental due to the finite-blocklength penalty induced by switching overhead. Nevertheless, the discrete problems \(\mathbf{P1}\)--\(\mathbf{P3}\) are solved directly over the feasible integer set, and no uniqueness of the discrete optimizer is assumed.
\end{remark}

\subsubsection{EE Trend}
The EE metric $\eta_{\rm EE}(x)=\bar{R}(x)/P_{\rm avg}(x)$ reflects the ratio between the achievable-rate gain and the average power consumption. As $N$ increases, $\bar{R}(N)$ first benefits from spatial diversity but later suffers from DoF saturation and finite-blocklength penalty. Meanwhile, $P_{\mathrm{avg}}(N)$ increases with the port-scanning overhead. This interaction explains the non-monotonic EE behavior observed in the numerical results. The EE-optimal port number is obtained by finite search over the feasible integer set.

\vspace{-2mm}
\subsection{Solution Algorithm and Complexity Analysis}
Let the feasible integer set be
\begin{equation}
\mathcal{N}=\{1,2,\ldots,N_{\max}\},
\qquad
N_{\max}=\left\lfloor\frac{L_{\mathrm{tot}}-L_{\min}}{\tau}\right\rfloor .
\end{equation}
The discrete problems $\mathbf{P1}$--$\mathbf{P3}$ are solved by evaluating the corresponding objective over all $N\in\mathcal{N}$. If multiple integer port numbers give the same optimal objective value, the smallest one is selected to reduce hardware complexity and switching overhead. The resulting finite-search procedure is summarized in Algorithm~\ref{alg:optimization}.

\begin{algorithm}[t]
\caption{Optimal Port Dimensioning Strategy}
\label{alg:optimization}
{\small
\begin{algorithmic}[1]
\Require System parameters $(L_{\mathrm{tot}},\tau,D,\bar{\gamma},W,P_{\mathrm{scan}},P_{\mathrm{active}})$, energy-retention factor $\eta$, and optimization mode
\Ensure Optimal number of ports $N^{\star}$

\State Set $N_{\max}\leftarrow \left\lfloor (L_{\mathrm{tot}}-L_{\min})/\tau \right\rfloor$
\State Set $\mathcal{N}\leftarrow \{1,2,\ldots,N_{\max}\}$

\For{each candidate $N\in\mathcal{N}$}
    \State Set $L(N)\leftarrow L_{\mathrm{tot}}-N\tau$
    \State Construct the Jakes covariance matrix $\mathbf{J}_N$ using \eqref{eq:1d_jakes}
    \State Compute and sort the eigenvalues $\lambda_1\geq\lambda_2\geq\cdots\geq\lambda_N$
    \State Determine $M(N)$ using the criterion in \eqref{eq:effective_rank}
    \State Generate $\Xi_{\mathbf{s}}=\sum_{j\in\mathbf{s}}\lambda_j^{-1}$ for $\mathbf{s}\in\mathcal{S}_k$, $k=1,\ldots,M(N)$

    \If{mode is Reliability}
        \State Set $\mathcal{O}_N\leftarrow\bar{\epsilon}(N)$ using \eqref{eq:closed_form_result}
    \ElsIf{mode is Throughput}
        \State Set $\mathcal{O}_N\leftarrow\bar{R}(N)$ using \eqref{eq:rate_closed_form}
    \Else
        \State Compute $P_{\mathrm{avg}}(N)$ using \eqref{eq:energy_model}
        \State Set $\mathcal{O}_N\leftarrow\bar{R}(N)/P_{\mathrm{avg}}(N)$
    \EndIf
\EndFor

\If{mode is Reliability}
    \State Set $N^{\star}\leftarrow \arg\min\limits_{N\in\mathcal{N}}\mathcal{O}_N$
\Else
    \State Set $N^{\star}\leftarrow \arg\max\limits_{N\in\mathcal{N}}\mathcal{O}_N$
\EndIf

\State \textbf{return} $N^{\star}$
\end{algorithmic}
}
\end{algorithm}

\subsection{Complexity Analysis}
The cost of Algorithm~\ref{alg:optimization} includes covariance construction, eigenvalue decomposition, effective-rank selection, subset-based metric evaluation, and finite search. For each candidate $N$, constructing $\mathbf{J}_N$ requires $\mathcal{O}(N^2)$ operations and eigenvalue decomposition requires $\mathcal{O}(N^3)$ operations. After $M(N)$ is determined by \eqref{eq:effective_rank}, the subset expansion in \eqref{eq:pdf_exact} contains $2^{M(N)}-1$ exponential terms. Hence, the dominant per-candidate complexity is $\mathcal{O}(N^3+2^{M(N)})$, and the total complexity over the feasible set is
\begin{equation}
\mathcal{O}\left(
\sum_{N\in\mathcal{N}}
\left(
N^3+2^{M(N)}
\right)
\right).
\end{equation}
This expression makes explicit that the exponential term depends on the effective spatial rank $M(N)$ rather than directly on the physical number of ports. Therefore, the analytical evaluation is efficient when $M(N)\ll N$, as commonly observed in compact FAS apertures with strong spatial correlation.
\section{Numerical Results and Discussion}\label{sec:numerical_results}
In this section, we validate the analytical framework and examine the impact of port switching delay on FAS-enabled HRLLC. Unless otherwise specified, the default parameters are $L_{\mathrm{tot}}=500$, $D=256$ bits, $W=5$, and $\tau=2$ channel uses per port. The value $\tau=2$ is interpreted as a normalized low-overhead benchmark under the quasi-static ideal-CSI setting, rather than as a standalone pilot duration for high-accuracy channel estimation. In all simulations, we enforce $L(N)=L_{\mathrm{tot}}-N\tau\geq L_{\min}=100$.

\begin{figure}[t]
\centering
\includegraphics[width=.7\columnwidth]{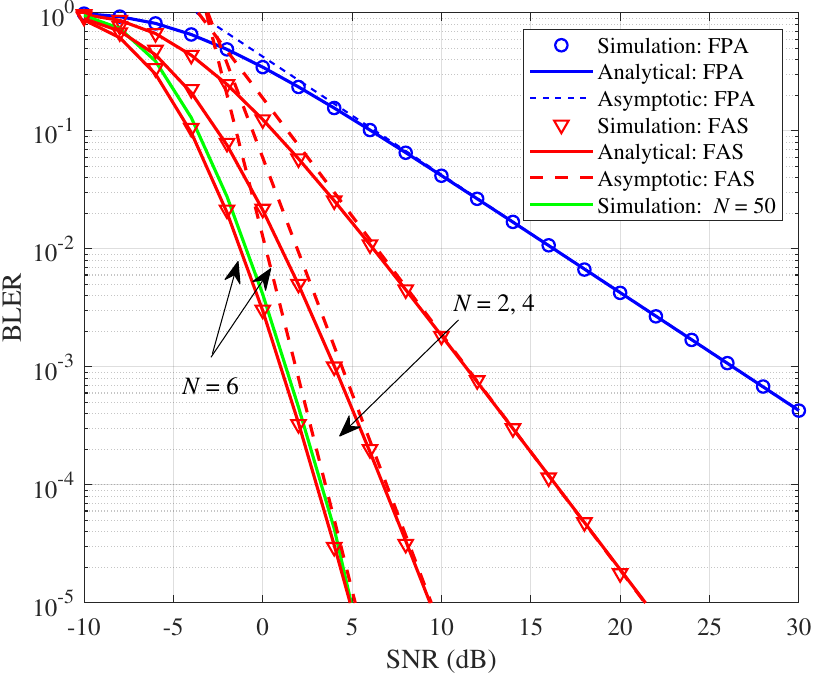}\\
\caption{Average BLER versus transmit SNR for varying number of ports $N$.}\label{Fig0}
\vspace{-2mm}
\end{figure}

\begin{figure}[t]
\centering
\includegraphics[width=.7\columnwidth]{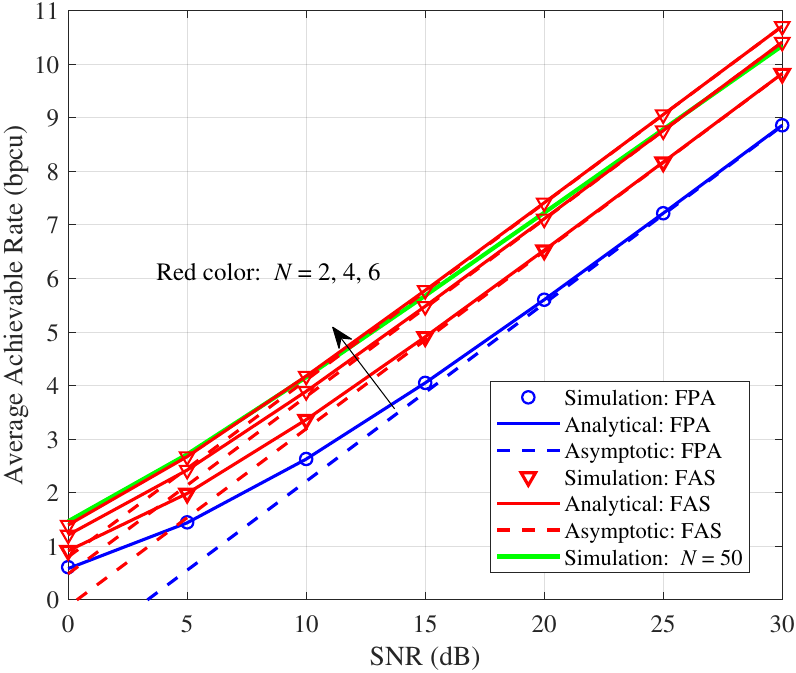}\\
\caption{Average achievable rate versus transmit SNR for varying number of ports $N$.}\label{Fig001}
\vspace{-2mm}
\end{figure}

\begin{figure}[t]
\centering
\includegraphics[width=.7\columnwidth]{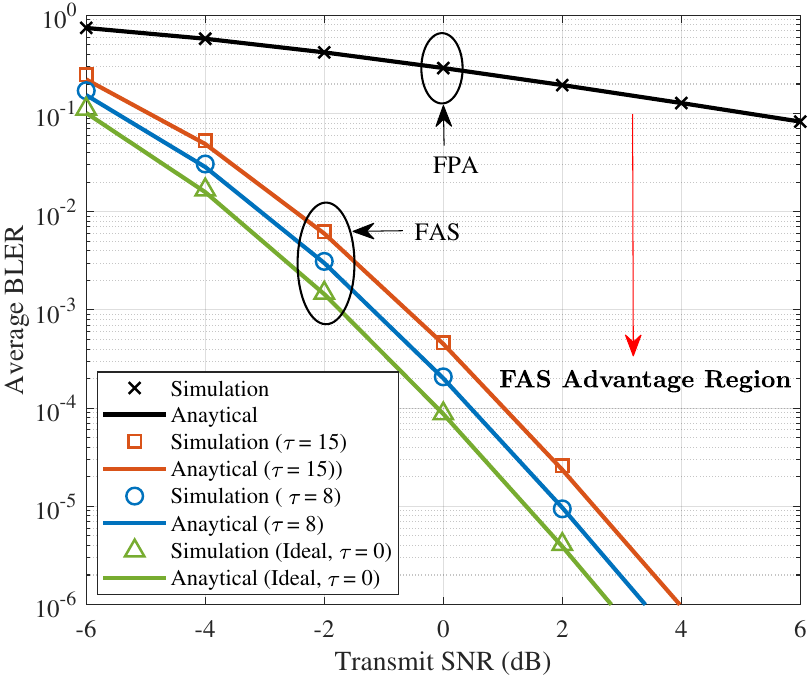}\\
\caption{Average BLER versus transmit SNR $\bar{\gamma}$ for the conventional FPA and the proposed FAS under different switching delay regimes.}\label{Fig3}
\vspace{-2mm}
\end{figure}

Fig.~\ref{Fig0} shows the average BLER versus the transmit SNR $\bar{\gamma}$ for various port dimensions $N \in \{1, 2, 4, 6, 50\}$, explicitly accounting for the practical port switching delay. Here, we consider a restricted normalized aperture of $W = 2$. As observed, the Monte Carlo simulation markers closely match the analytical approximations. In the high-SNR region, the simulation, analytical, and asymptotic curves converge into a single trajectory, validating our derivation of the diversity order $G_d$. Notably, the FAS-enabled system significantly outperforms the FPA by providing a massive reliability buffer through spatial diversity. Initially, reliability improves as $N$ increases from $2$ to $6$ because a higher port density allows the system to harvest higher diversity gain $M$. However, a reversal occurs at $N=50$, where performance is noticeably inferior to $N=6$. While $N=50$ offers higher diversity, the linear depletion of the effective blocklength $L(N) = L_{\mathrm{tot}} - N\tau$ significantly increases the coding penalty $\mathcal{K}$. In the ultra-low latency regime, the penalty from reduced transmission time outweighs the marginal diversity gains from dense ports, which is consistent with the high-SNR reliability trade-off explained in Proposition~\ref{prop:convexity_bler}.

Fig.~\ref{Fig001} compares the average achievable rate against the transmit SNR $\bar{\gamma}$ to evaluate throughput under finite blocklength constraints. {In this scenario, the normalized aperture is set to $W = 2$, and the port switching overhead is reduced to $\tau = 0.5$ channel uses per port to highlight the throughput potential.} We evaluate the system for $N \in \{1, 2, 4, 6, 50\}$. The results show that simulation markers closely match the closed-form analytical approximations. At high SNRs, the simulation, analytical, and asymptotic results converge, validating the spatial diversity gain offset $\mathcal{S}_{\text{div}}(N)$ derived in our asymptotic analysis. Notably, while the achievable rate initially improves as $N$ increases from $2$ to $6$ due to enhanced spatial sampling, a performance reversal occurs at $N=50$. Here, the rate gains saturate or diminish because the effective blocklength $L(N)$ decreases linearly with $N$. This reduction in transmission time amplifies the finite blocklength penalty, eventually outweighing the diversity benefits, which is consistent with the high-SNR throughput trend discussed in Section~\ref{subsec:high_snr_throughput_trend}.

\begin{figure}[t]
\centering
\includegraphics[width=.7\columnwidth]{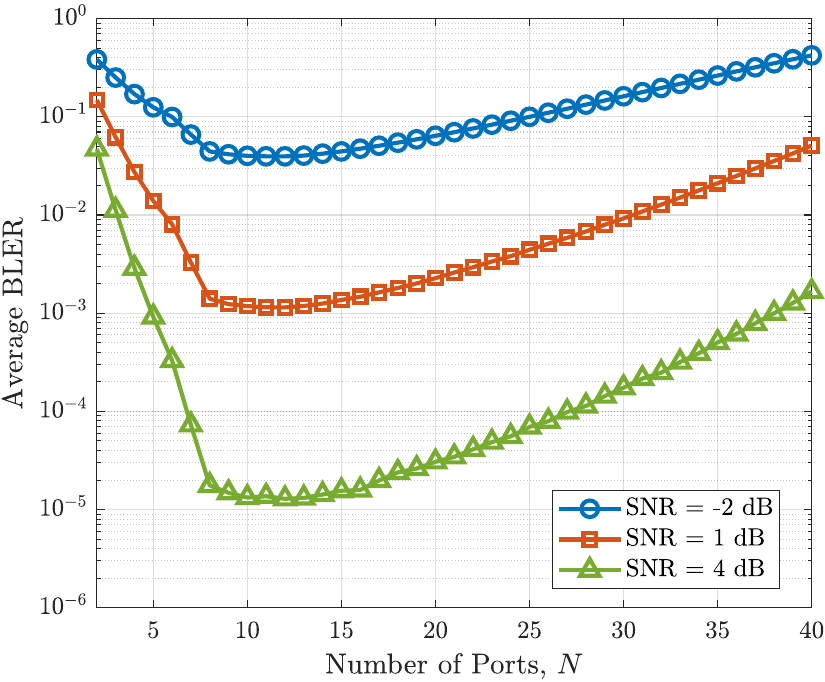}\\
\caption{Average BLER versus the number of ports $N$ for varying transmit SNRs.}\label{Fig1}
\vspace{-2mm}
\end{figure}

\begin{figure}[t]
\centering
\includegraphics[width=.7\columnwidth]{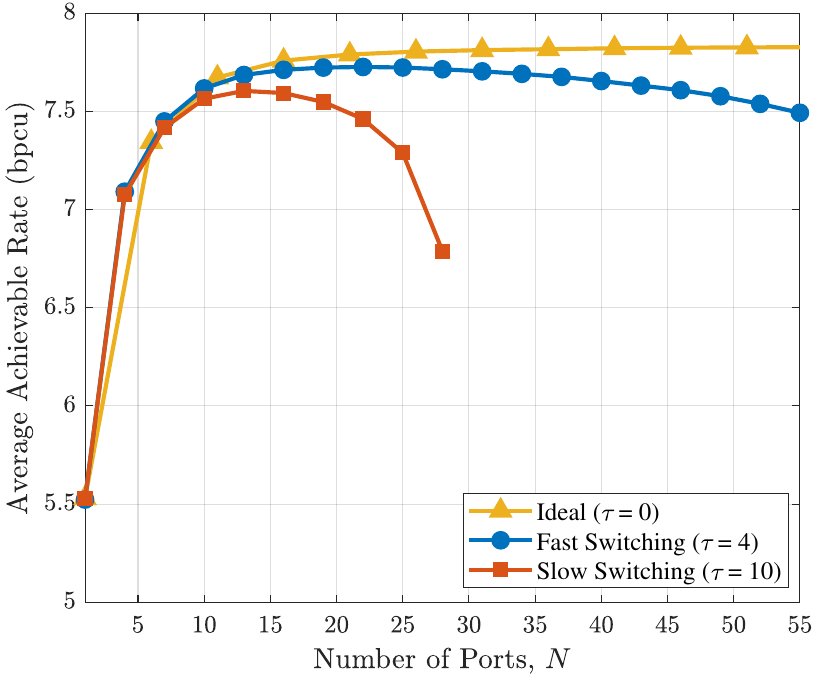}\\
\caption{Average achievable rate versus the number of ports $N$ for different switching delay constraints.}\label{Fig2}
\vspace{-2mm}
\end{figure}

\begin{figure}[t]
\centering
\includegraphics[width=.7\columnwidth]{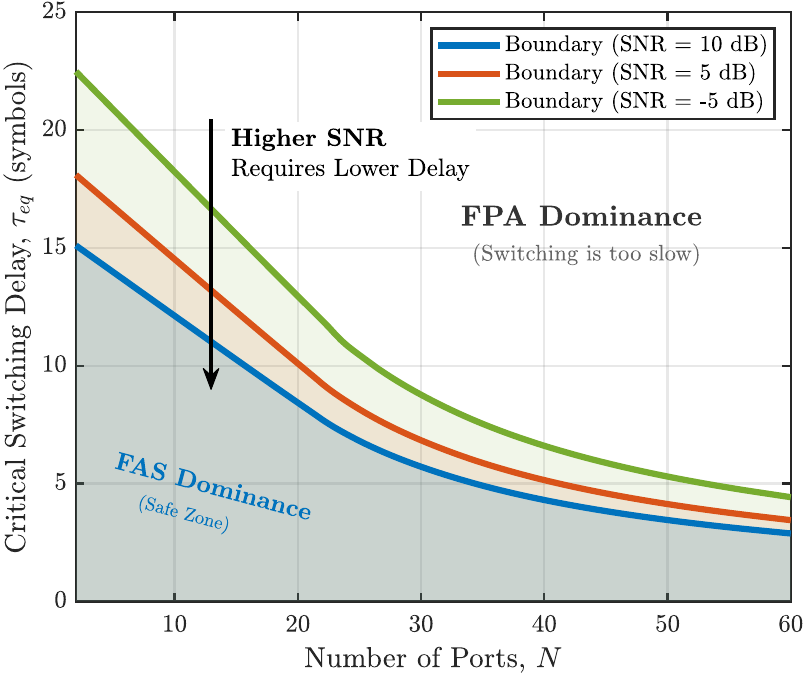}\\
\caption{The critical switching delay threshold $\tau_{\rm eq}$ versus the number of ports $N$ for varying transmit SNRs.}\label{Fig4}
\vspace{-2mm}
\end{figure}

Fig.~\ref{Fig3} compares the average BLER of the proposed FAS against the FPA. For the results in this figure, the latency budget is extended to $L_{\mathrm{tot}} = 600$, and the aperture is set to $W = 3$. We compare a reference FPA ($N=1$) with an 8-port FAS under ideal ($\tau=0$), practical ($\tau=8$), and slow ($\tau=15$) switching cases. The results validate a clear ``FAS Advantage Region" where FAS provides a much steeper error decay slope than FPA due to high-order diversity. Notably, even with a large delay of $\tau = 15$, the FAS still performs better than the FPA in the tested SNR range. However, the gaps between the FAS curves demonstrate that higher switching delays incur a reliability penalty because the effective blocklength $L_{\mathrm{eff}}$ is shortened, reducing coding efficiency.

\begin{figure}[t]
\centering
\includegraphics[width=.7\columnwidth]{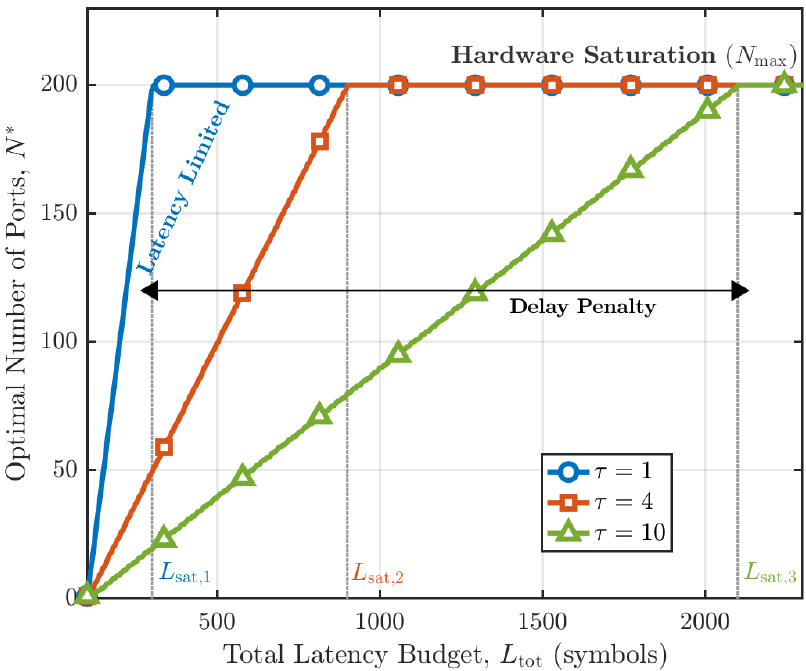}\\
\caption{The optimal number of ports $N^{\star}$ versus the total latency budget $L_{tot}$ for different switching delay values.}\label{Fig5}
\vspace{-2mm}
\end{figure}

Fig.~\ref{Fig1} studies the reliability-latency trade-off under {more stringent constraints: $L_{\mathrm{tot}} = 300$, $D = 200$, $W = 3$, and $\tau = 4$.} As observed, the BLER exhibits a non-monotonic, U-shaped behavior with respect to $N$. Initially, reliability improves as $N$ increases (from $2$ to $\approx 10$) due to enhanced spatial diversity. However, a reversal occurs as $N$ grows further, driven by the linear depletion of the effective blocklength. In this overhead-dominated regime, the reduction in channel uses increases the dispersion penalty, outweighing marginal diversity gains. As such, a finite search over the feasible integer set identifies the best port count for each SNR level.

\begin{figure}[t]
\centering
\includegraphics[width=.7\columnwidth]{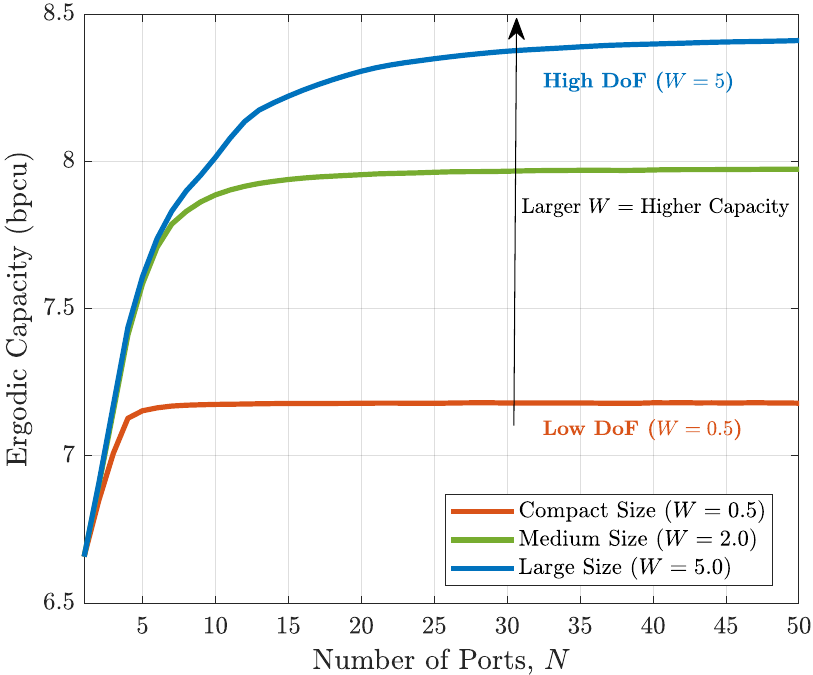}\\
\caption{Ergodic capacity versus the number of ports $N$ for different normalized aperture sizes.}\label{Fig6}
\vspace{-2mm}
\end{figure}

\begin{figure}[t]
\centering
\includegraphics[width=.7\columnwidth]{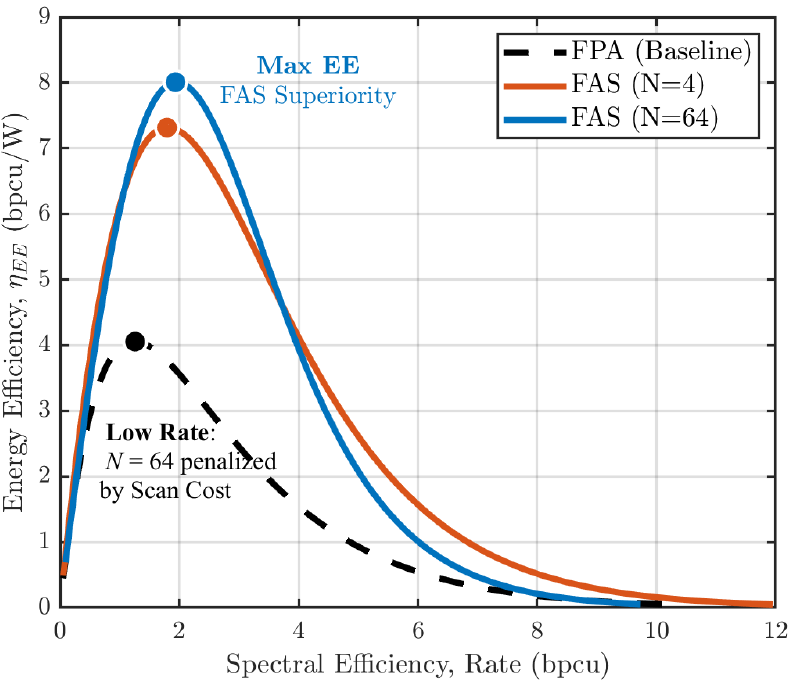}\\
\caption{Energy efficiency versus spectral efficiency trade-off.}\label{Fig7}
\vspace{-2mm}
\end{figure}

\begin{figure}[t]
\centering
\includegraphics[width=.7\columnwidth]{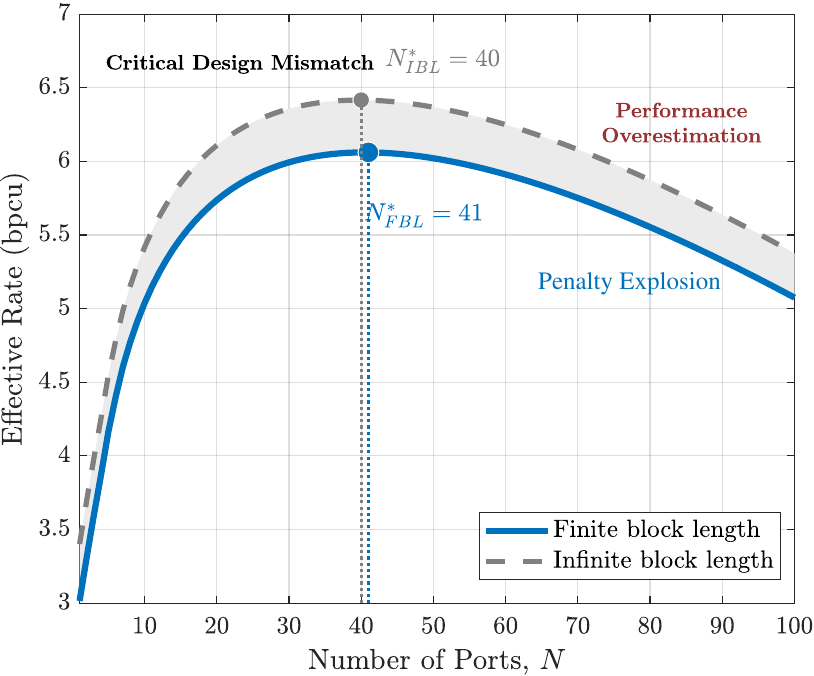}\\
\caption{Effective achievable rate versus the number of ports $N$ under finite-blocklength and infinite-blocklength regimes.}\label{Fig11}
\vspace{-2mm}
\end{figure}

\begin{figure}[t]
\centering
\includegraphics[width=.7\columnwidth]{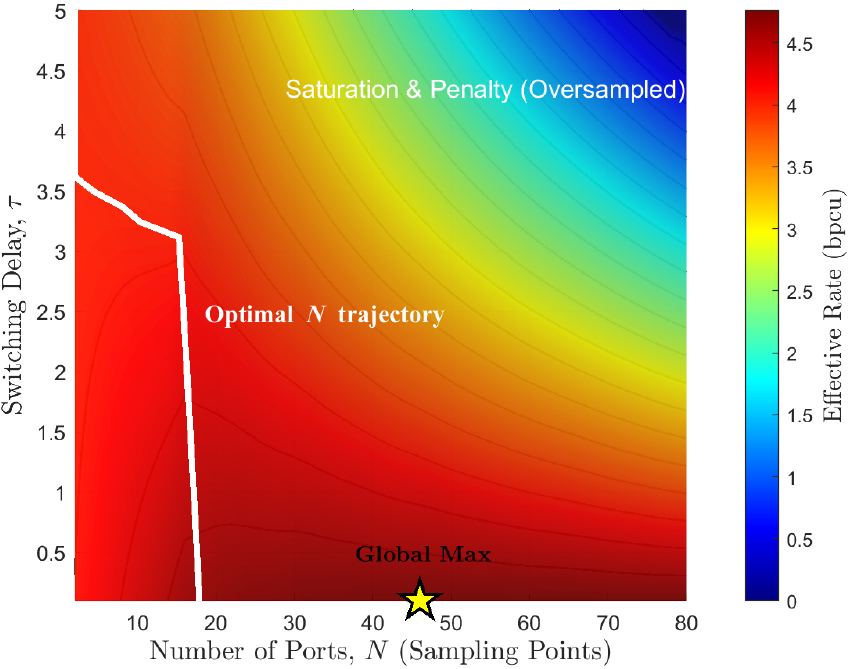}\\
\caption{Joint impact of port dimensioning and switching delay on effective rate.}\label{Fig12}
\vspace{-2mm}
\end{figure}

Fig.~\ref{Fig2} shows the average achievable rate versus $N$ under varying switching delays. Here, we adopt a reduced latency budget $L_{\mathrm{tot}}=300$ and a normalized aperture $W=3$ at an SNR of $20$ dB. Throughput exhibits a distinct non-monotonic behavior for all practical cases with $\tau>0$, which is consistent with the high-SNR throughput trend discussed in Section~\ref{subsec:high_snr_throughput_trend}. While increasing $N$ helps harvest more spatial modes, it also depletes the effective blocklength and amplifies the finite-blocklength penalty. Furthermore, the best port count is highly sensitive to the switching delay; a transition from fast switching ($\tau=4$) to slow switching ($\tau=10$) shifts the selected port count to the left. The ideal case with $\tau=0$ serves as a monotonic upper bound, highlighting the temporal cost of port scanning.

In Fig.~\ref{Fig4}, results are provided to show the critical switching delay limit $\tau_{\rm eq}$. The latency budget is set to $L_{\mathrm{tot}} = 400$, evaluated at SNRs of $-5$, $5$, and $10$ dB. The results define a ``FAS Dominance" zone where the diversity gain outweighs the time loss. The delay limit $\tau_{\rm eq}$ drops as $N$ increases, indicating that systems with more ports are more sensitive to switching costs. Additionally, higher SNRs enlarge the feasible-delay region, implying that stronger diversity gains can tolerate larger switching delays before the finite-blocklength penalty dominates.

Fig.~\ref{Fig5} plots the optimal number of ports $N^{\star}$ as a function of $L_{\mathrm{tot}}$ for varying $\tau$. The simulation evaluates Algorithm \ref{alg:optimization} under a target BLER of $10^{-5}$ at SNR $= 20$ dB. As observed, $N^{\star}$ initially increases with $L_{\mathrm{tot}}$ but then reaches a plateau, indicating a hardware saturation limit. The growth rate of $N^{\star}$ is highly sensitive to switching speed; with fast switching ($\tau=1$), $N^{\star}$ grows rapidly and saturates earlier than in the slow switching case ($\tau=10$). This confirms that Algorithm \ref{alg:optimization} successfully adapts port dimensioning to available temporal resources, balancing the diversity-multiplexing trade-off.

Fig.~\ref{Fig6} illustrates the ergodic capacity versus $N$ for different normalized aperture sizes $W \in \{0.5, 2.0, 5.0\}$ at an SNR of $20$ dB to highlight the DoF saturation effect. Initially, capacity increases rapidly with $N$, but this gain plateaus once the DoF limit is reached. For a compact size ($W=0.5$), capacity saturates early ($ N\approx2$) due to high spatial correlation. In contrast, a large aperture ($W=5.0$) supports a much higher DoF, allowing growth until a larger $N$. This shows that the spatial dimension $W$ is the primary resource dictating capacity.

Fig.~\ref{Fig7} illustrates the EE versus spectral efficiency. Using the default system configuration, we compare FPA with FAS for $N=4$ and $N=64$. The results verify the non-monotonic EE trade-off. In the low-rate regime, $N=64$ shows lower EE due to dominant scanning costs ($N\tau P_{\rm scan}$). However, as the spectral efficiency target increases, the system enters a diversity-dominant regime where $N=64$ outperforms $N=4$ by reducing the required transmit SNR. At very high data rates, the performance of $N=64$ degrades sharply because the large switching overhead consumes the effective blocklength. Thus, $N=4$ becomes more efficient again at the highest rates.

Fig.~\ref{Fig11} compares the effective achievable rate under finite blocklength and infinite blocklength regimes. The simulation is configured with a restricted latency budget $L_{\mathrm{tot}} = 250$, aperture $W = 1.0$, switching delay $\tau = 1.0$, and SNR $= 10$ dB. The infinite blocklength model consistently overestimates performance by ignoring the HRLLC channel dispersion penalty. The two models may lead to different port-dimensioning decisions; for example, the infinite-blocklength model suggests $N^{\star}=40$, whereas the finite-blocklength model identifies $N^{\star}=41$. As $N$ increases, the finite blocklength curve experiences a ``penalty explosion" where the decreasing effective blocklength amplifies the dispersion penalty, widening the gap between the two models.

Finally, Fig.~\ref{Fig12} shows the joint impact of port dimensioning and switching delay on system performance,  evaluated at an SNR of $10$ dB with a latency budget of $L_{\mathrm{tot}}=400$. The results show a clear non-monotonic surface with several key physical insights. First, for small $N$, the rate increases rapidly because the system captures more spatial diversity. However, once $N$ exceeds the physical DoF limit, the rate saturates because the ports become highly correlated. Meanwhile, increasing the switching delay $\tau$ significantly reduces the rate by shrinking the effective blocklength. Therefore, the optimal $N$ follows a downward trajectory as $\tau$ grows, showing a trade-off between diversity gain and time overhead. Also, the largest rates are achieved when the switching delay is small, where the system can best exploit spatial freedom without excessive scanning loss.

\vspace{-2mm}
\section{Conclusion}\label{sec:conclusion}
This paper investigated delay-aware FAS for 6G HRLLC under finite blocklength operation. By jointly considering spatial correlation, port switching delay, and short-packet coding, we showed that increasing the number of ports improves spatial diversity but reduces the effective blocklength. Based on an effective spatial-DoF representation, closed-form analytical approximations were derived for the average BLER and achievable rate. The analysis revealed a reliability-rate-latency trade-off and characterized switching-delay thresholds for FAS gains over FPAs. The discrete port-dimensioning problems were solved by finite search over the feasible integer set. Numerical results confirmed that FAS can provide HRLLC gains when the switching delay is sufficiently small. A limitation is that channel estimation errors, CSI aging, and multi-antenna BS transmission are not explicitly modeled; robust FAS-HRLLC design under imperfect/outdated CSI and joint BS beamforming with FAS port selection are left for future work.

\end{document}